\documentclass[11pt, a4paper, reqno]{amsart}

\usepackage[utf8]{inputenc}

\usepackage[letterpaper,hmargin=1.0in,vmargin=1.0in]{geometry}
\usepackage{amsmath,amsthm,amsfonts,amssymb,mathrsfs,stmaryrd,bigints}
\usepackage{dsfont}
\usepackage{mathtools}
\usepackage[usenames]{color}
\usepackage{MnSymbol}
\usepackage[colorlinks=true,linkcolor=blue]{hyperref}

\usepackage{tikz}
\usetikzlibrary{decorations.pathreplacing,calligraphy}
\usepackage{comment}
\usepackage{float}
\usepackage{epstopdf}
\usepackage{amssymb}
\usepackage{setspace}
\usepackage{enumerate}
\usepackage{bigstrut}
\usepackage{multirow}
\usepackage{mathtools,xparse}
\usepackage{mathrsfs}
\usepackage{hyperref}

\usepackage [autostyle, english = american]{csquotes}
\newtheorem*{lemma*}{Lemma}
\newtheorem{theorem}{Theorem}[section]
\newtheorem{corollary}[theorem]{Corollary}
\newtheorem{lemma}[theorem]{Lemma}
\newtheorem{proposition}[theorem]{Proposition}
\theoremstyle{definition}
\newtheorem{definition}{Definition}
\newtheorem{remark}[theorem]{Remark}

\allowdisplaybreaks
\usepackage{chngcntr}

\newcommand{\ent}{h}

\newcommand{\tmix}{t_{\rm{mix}}}

\DeclareMathOperator{\diam}{diam}

\renewcommand{\setminus}{\backslash}

\newtheorem{maintheorem}{Theorem}

\def\Ld L
\def\L{\mathbb L}

\def\been{\begin{enumerate}}
\def\bee{\begin{example}}
\def\beit{\begin{itemize}}
\def\bel{\begin{lemma}}
\def\bepr{\begin{proposition}}
\def\bep{\begin{proof}}
\def\bet{\begin{theorem}}
\def\bec{\begin{corollary}}
\def\enc{\end{corollary}}

\def\de\de 
\def\diam{\ms{diam}}

\def\enen{\end{enumerate}}
\def\ene{\end{example}}
\def\enit{\end{itemize}}
\def\enl{\end{lemma}}
\def\enpr{\end{proposition}}
\def\enp{\end{proof}}
\def\ent{\end{theorem}}

\def\ms{\mathsf}

\def\de{\delta}

\def\to \infty{\uparrow\infty}

\def\bef{\begin{figure}[!h]}
\def\enf{\end{figure}}

\usepackage{kotex}

\begin{document}

\title{Exponentially slow Mixing arising from Entropic Repulsion in $p$-SOS model}

\author{Seokun Choi}
\begin{abstract}

We investigate the Glauber dynamics of the generalized (2+1)-dimensional $p$-SOS model($1<p<\infty$) under a hard floor constraint. This setting induces entropic repulsion: the integer-valued interface height is forced to remain above the wall and consequently rises to a typical height $H(p,L)$ that depends on both the parameter $p$ and the lattice size $L$. The phenomenon of entropic repulsion has been extensively studied in a variety of random interface models. In the classical SOS model ($p=1$), \cite{caputo2016scaling, caputo2014dynamics} derived an exponential lower bound for the mixing time, demonstrating that the Glauber dynamics mixes only after an exponentially long time in the low-temperature regime (large $\beta$, the inverse temperature). However, beyond this case, no rigorous lower bounds were previously known: even for the widely studied Discrete Gaussian model ($p=2$), the metastable slowdown predicted by the entropic repulsion picture had remained an open problem. On the equilibrium side, \cite{lubetzky2016harmonic} obtained sharp large-deviation principles and precise estimates of typical and maximal heights for all $1\le p<\infty$, but the dynamical consequences of these results had not been established.

Our main contribution is to close this gap by proving that exponentially slow(stretched-exponential) mixing arising from entropic repulsion persists throughout the regime $1<p<\infty$. Specifically, we establish an exponential lower bound, showing that the mixing time satisfies $\tau_{\mathrm{mix}}\ge \exp{\left(cL^{1-o(1)}\right)}$ for some $c>0$ depending on $p$ and $\beta$. In addition, we provide a refined metastability analysis, proving that the hitting time of an intermediate level $aH(p,L)$ is at least $\exp{\left(cL^{a^{d(p)}-o(1)}\right)}$, where $0<a<1$ and $d(p)$ is a positive function depending on $p$. The proof relies on an extension of the classical Peierls-type contour estimates, originally developed for $p=1$, to the nonlinear $p$-SOS setting. Taken together, these results demonstrate that entropic repulsion induces uniformly slow mixing across the entire $p$-SOS family, thereby extending a phenomenon that had previously been established only for $p=1$.

\end{abstract}
\address[Seokun Choi]{Department of Mathematical Sciences, KAIST, South Korea}
\email{seokunchoi@kaist.ac.kr}
\subjclass[2010]{}
\maketitle
\tableofcontents

\section{Introduction}

\subsection{Background}
The Solid-On-Solid (SOS) model, originally introduced by Temperley \cite{temperley1952statistical} in 1952 and also known as the Onsager-Temperley sheet, was proposed to describe the microscopic behavior of the crystal interface. In modern statistical mechanics, this model is viewed as an effective description of the interface separating the positive and negative phases phases in the $(d+1)$-dimensional Ising model, particularly in the low-temperature regime. In this setting, the SOS interface approximates the Ising interface by suppressing overhangs and bulk fluctuations, which are energetically costly at low temperatures.

We define the SOS model under general boundary conditions as follows. For a finite domain $\Lambda \subset \mathbb{Z}^2$, consider integer-valued height functions 
\[
  \eta:\Lambda \rightarrow \mathbb{Z}, \qquad \eta \equiv \phi \;\; \text{on } \Lambda^c,
\]
where $\phi:\Lambda^c \rightarrow \mathbb{Z}$ is a fixed boundary condition. 
The Hamiltonian is given by
\[
  H^\phi_\Lambda(\eta) \;=\; \sum_{\substack{x\in \Lambda,\, y\in \mathbb{Z}^2 \\ |x-y|=1}} 
  \big|\eta(x) - \xi(y)\big|,
  \qquad 
  \xi(y) \;=\; 
  \begin{cases}
    \eta(y), & y \in \Lambda,\\
    \phi(y), & y \in \Lambda^c,
  \end{cases}
\]
which includes both the interactions inside $\Lambda$ and those with the prescribed boundary condition. 
The corresponding Gibbs measure is then
\[
  \hat{\pi}_\Lambda^\phi(\eta) \;=\; \frac{1}{Z^{\phi}_\Lambda}\,
  \exp\!\big(-\beta H^\phi_\Lambda(\eta)\big),
\]
where $\beta>0$ is the inverse temperature and $Z^\phi_\Lambda$ is the partition function ensuring normalization. In the special case of zero boundary condition $\phi \equiv 0$, we denote the Gibbs measure by $\hat{\pi}_\Lambda^0(\eta),$ and for convenience we shall simply write $\hat{\pi}_\Lambda(\eta)$.

The two-dimensional case $d=2$ plays a distinguished role, being the unique dimension where a roughening transition takes place. 
For $d=1$, the interface stays non-localized for all inverse temperatures $\beta>0$, 
implying that the average height at a given site diverges as the system size increases 
\cite{fisher1984walks, temperley1956combinatorial}.
In contrast, for $d\geq 3$, Peierls-type arguments show that the interface is localized 
\cite{bricmont1982surface}, with height fluctuations uniformly bounded in expectation. 
In two dimensions, however, a transition occurs between delocalized and localized 
behavior \cite{brandenberger1982decay, gallavotti2007some}: at high temperatures the 
interface becomes rough, while at low temperatures it is localized. 
This critical phenomenon was first rigorously analyzed in the foundational works of 
Fröhlich and Spencer \cite{frohlich1981kosterlitz, frohlich1983berezinskii}.

To capture a broader range of interface behaviors, the classical SOS model 
\cite{caputo2016scaling, lubetzky2016harmonic} can be extended to the so-called $p$-SOS model, defined for $p\geq 1$. Configurations follow the SOS model, that is, height functions taking integer values on $\Lambda \subset \mathbb{Z}^2$ together with a specified boundary condition $\phi:\Lambda^c \rightarrow \mathbb{Z}$.
The Hamiltonian is given by
\[
  H_\Lambda^{p,\phi}(\eta) \;=\; 
  \sum_{\substack{x \in \Lambda,\, y \in \mathbb{Z}^2 \\ |x-y|=1}}
  \big|\eta(x) - \xi(y)\big|^p,
  \qquad 
  \xi(y) \;=\;
  \begin{cases}
    \eta(y), & y \in \Lambda,\\
    \phi(y), & y \in \Lambda^c,
  \end{cases}
\]
and the corresponding Gibbs measure is
\begin{align}\label{defgibb}
    \hat{\pi}_\Lambda^{p,\phi}(\eta) \;=\; \frac{1}{Z_\Lambda^{p,\phi}}\,
  \exp\!\big(-\beta H_\Lambda^{p,\phi}(\eta)\big),    
\end{align}
where $\beta>0$ is the inverse temperature, $p \in [1,\infty)$ is fixed, 
and $Z_\Lambda^{p,\phi}$ is the partition function ensuring normalization. In the special case of zero boundary condition $\phi \equiv 0$, we write $\hat{\pi}_\Lambda^{p,0}(\eta)$ and, for convenience, simply denote it by $\hat{\pi}_\Lambda^p(\eta)$.

This framework includes two notable special cases: when $p = 1$, the model reduces to the classical SOS model, while $p = 2$ corresponds to the Discrete Gaussian(DG) model. The parameter $p$ thus interpolates between models with different types of gradient penalization, allowing for a unified approach to interface fluctuation analysis. In this work, we aim to extend certain results previously established for the case $p = 1$ to the broader range $1 \leq p < \infty$. Recent developments show a growing interest in such generalized models \cite{bauerschmidt2024discrete1, bauerschmidt2024discrete2, caputo2017entropic, lubetzky2016harmonic}, as they offer a flexible framework for understanding a wide class of random interfaces.

One of the central phenomena in random interface (surface) models with hard constraints, for instance when a floor prevents penetration, is the \textit{entropic repulsion}.
 When the interface is forced to stay above a fixed wall, the entropic term in the Gibbs distribution pushes the interface upward, in opposition to the energy term that favors minimizing gradients.
To model this situation mathematically, one considers a sequence of expanding domains $\Lambda_L :=\{1, 2, \cdots, L\}^2\subset\mathbb{Z}^2$ and studies the field $\eta$ conditioned to be nonnegative on $\Lambda_L$. That is, one defines the conditional measure
\[
\bar{\pi}_{\Lambda_L}^{p}(\cdot) := \hat{\pi}_{\Lambda_L}^p(\cdot \mid \Omega_{\Lambda_L}^+), \quad \text{where } \Omega_{\Lambda_L}^+ = \{\eta : \eta_x \geq 0 \text{ for all } x \in {\Lambda_L}\}.
\]
This describes a random interface constrained to remain above a wall. A notable feature of this conditioning is that the interface’s roughness remains essentially unchanged, while the entire interface shifts upward, keeping a finite distance from the wall. This behavior reflects the interface’s tendency to preserve its freedom to fluctuate, which lies at the core of the entropic repulsion phenomenon. From a mathematical standpoint, entropic repulsion has been rigorously investigated in various settings. In particular, the phenomenon was first rigorously analyzed in the context of SOS, 
Discrete Gaussian, and related random interface models by Bricmont, El Mellouki, 
and Fröhlich (1986) \cite{bricmont1986random}, while subsequent works by Deuschel, 
Giacomin, and others \cite{bolthausen2001entropic, bolthausen1995entropic} 
investigated it in the setting of the Gaussian free field (GFF). Such behavior is observed not only in the classical SOS model but also in generalized $p$-SOS models and continuous deformations thereof. 

Building on this foundation, more recent research has focused on sharper quantitative properties of the interface. In particular, the study of one-point large deviations and the equilibrium behavior in $p$-SOS models has become increasingly active \cite{lubetzky2016harmonic}. In this context, it is natural to extend the investigation of entropic repulsion to the $p$-SOS setting, where a first step is to ensure the existence of the infinite-volume Gibbs measure. For sufficiently large $\beta$ (low temperature), $\hat{\pi}^p$ is well-defined and unique, by standard low-temperature methods such as contour expansion and the Pirogov–Sinai theory.
In the SOS ($p=1$) and DG ($p=2$) cases, this was proved in \cite{bricmont1986random} (Section~2.1), 
where it is shown that the height distribution has an exponential tail 
and correlations decay exponentially. 
The same reasoning extends to general $p \geq 1$.

In the entropic repulsion regime, the $p$-SOS interface concentrates around a well-defined typical height, with nearly all heights lying close to the same level. The \emph{typical height under entropic repulsion}, which depends on both the parameter $p$ and the system size $L$, is defined based on ~\cite{lubetzky2016harmonic} by
\begin{equation}\label{typicalheight}
H(p, L) := \max \left\{ h \in \mathbb{N} : \hat{\pi}^p(\eta_0 \ge h) \ge \frac{5\beta}{L} \right\}.
\end{equation}
Here and throughout, we write $a \asymp b$ if both $a = O(b)$ and $b = O(a)$. The large deviation rate of $\hat{\pi}^p(\eta_0 \ge h)$ and the corresponding 
typical height $H(p,L)$ for different values of $p$ are summarized in Table~\ref{tab:HpL}, based on the results of 
\cite{caputo2016scaling, caputo2014dynamics, lubetzky2016harmonic}. While equilibrium properties of entropic repulsion are by now well understood for both the SOS model ($p=1$) and its generalizations to $1\le p\leq\infty$ (see e.g. \cite{lubetzky2016harmonic}), the dynamical behavior has remained much less explored. In particular, the phenomenon of slow mixing arising from entropic repulsion was established by Caputo et al.
 \cite{caputo2014dynamics}, who established an exponentially slow mixing time in the SOS case $p=1$. Beyond this single case, no analogous lower bound was known, not even for the classical Discrete Gaussian model ($p=2$). In this work, we focus on extending the lower bound on the mixing time under entropic repulsion, previously obtained for the SOS model in~\cite{caputo2014dynamics}, to the more general $p$-SOS model. Since our work builds directly on the framework of \cite{caputo2014dynamics}, several of our formulations and terminology follow their presentation, with adaptations to the generalized $p$-SOS setting.

\begin{table}[H]
\centering
\caption{Large deviation rate of $\hat{\pi}^p(\eta_0 \ge h)$ and the corresponding typical height $H(p,L)$ for different values of $p$. 
All asymptotics follow from the cited works. 
Here $c_p$ denotes a constant depending only on $p$, and $\varepsilon_\beta \rightarrow 0$ as $\beta \rightarrow \infty$ and we write $a\asymp b$ if there exist constants $c_1,c_2>0$
such that $c_1 b \le a \le c_2 b$.}
\label{tab:HpL}
\begin{tabular}{|c|c|c|}
\hline
$p$ & $-\log \hat{\pi}^p(\eta_0 \ge h)$ & $H(p, L)$ \\
\hline
$1$ 
& $4\beta h + \varepsilon_\beta$
\;{\scriptsize\cite[Proposition 3.9]{caputo2014dynamics}}
& $\left\lfloor \frac{1}{4\beta}\log L \right\rfloor$
\;{\scriptsize\cite[Theorem 1]{caputo2016scaling}} \\
\hline
$1 < p < 2$ 
& $(c_p \beta + o(1))\, h^p$
\;{\scriptsize\cite[Theorem 5.1]{lubetzky2016harmonic}}
& $\left( \frac{1+o(1)}{c_p\beta}\log L \right)^{\frac{1}{p}+o(1)}$
\;{\scriptsize\cite[Corollary 5.2]{lubetzky2016harmonic}} \\
\hline
$p = 2$ 
& $(2\pi\beta + o(1))\, \dfrac{h^2}{\log h}$
\;{\scriptsize\cite[Theorem 2]{lubetzky2016harmonic}}
& $\left( \frac{1+o(1)}{4\pi\beta}\log L\log\log L \right)^{\frac{1}{2}+o(1)}$
\;{\scriptsize\cite[Theorem 3.1]{lubetzky2016harmonic}} \\
\hline
$2 < p < \infty$ 
& $\asymp \beta h^2$
\;{\scriptsize\cite[Theorem 5.5]{lubetzky2016harmonic}}
& $\asymp \left( \frac{\log L}{\beta} \right)^{\frac{1}{2}}$
\;{\scriptsize\cite[Corollary 5.6]{lubetzky2016harmonic}} \\
\hline
\end{tabular}
\end{table}

\subsection{$p$-SOS Model, Associated Gibbs Measures, and Glauber Dynamics}
\label{sec-model}
We next set up the Glauber dynamics (also known as heat bath dynamics or the Gibbs sampler) for the $p$-SOS model in $(2+1)$ dimensions
, the process we shall analyze in this work. For clarity, we also collect the Gibbs measure notation that will be used throughout.
Recall that $\Lambda_L := \{1, \ldots, L\}^2 \subset \mathbb{Z}^2$ and let $\partial \Lambda_L$ denote its outer boundary. 
Given a boundary condition $\phi : \partial \Lambda_L \rightarrow \mathbb{Z}$, 
the Hamiltonian $H_{\Lambda_L}^{p,\phi}$ and the corresponding Gibbs measure 
$\hat{\pi}_{\Lambda_L}^{p,\phi}$ were introduced above.  
In addition to the unconstrained Gibbs measure $\hat{\pi}_{\Lambda_L}^{p,\phi}$, we shall work with two constrained versions:
\begin{itemize}
    \item[(i)] With the constraint $\eta_x \ge 0$ for all $x \in \Lambda_L$ (no upper bound, i.e.\ $n_+=\infty$),
    supported on
    \[
      \Omega^+_{\Lambda_L} := \{ \eta \in \mathbb{Z}^{\Lambda_L} : \eta_x \ge 0 \;\; \forall x \in \Lambda_L \},
    \]
    we denote the corresponding Gibbs measure by $\bar{\pi}_{\Lambda_L}^{p,\phi}$.

    \item[(ii)] With the restriction $0 \le \eta_x \le n_+$ for all $x \in \Lambda_L$, where $n_+ \in \mathbb{N}$ and $n_+ \ge \log L$,
    supported on
    \[
      \Omega^+_{\Lambda_L,n_+} := \{ \eta \in \mathbb{Z}^{\Lambda_L} : 0 \le \eta_x \le n_+ \;\; \forall x \in \Lambda_L \},
    \]
    we denote the corresponding Gibbs measure by $\pi_{\Lambda_L}^{p,\phi}$.
\end{itemize}

We next define the Glauber dynamics as a discrete-time Markov chain $(\eta(t))_{t \ge 0}$ 
on $\Omega^+_{\Lambda_L,n_+}$. At each step:
\begin{itemize}
  \item A site $x \in \Lambda_L$ is chosen uniformly at random;
  \item The height $\eta_x$ is resampled from 
  $\pi_{\Lambda_L}^{p,\phi}(\cdot \mid \eta_{\Lambda_L \setminus \{x\}})$.
\end{itemize}

Let $\mathbb{P}_\nu$ denote the law of the chain with initial distribution $\nu$, 
and let $\mu_t^\nu$ be its distribution at time $t$. 
When starting from a fixed configuration $\eta$, we write $\mu_t^\eta$. 
This Markov chain is reversible with respect to $\pi_{\Lambda_L}^{p,\phi}$. 
The total variation mixing time is defined as
\[
  t_{\mathrm{mix}} := \inf \Big\{ t > 0 : \max_{\eta \in \Omega^+_{\Lambda_L,n_+}} 
  \| \mu_t^\eta - \pi_{\Lambda_L}^{p,\phi} \|_{\mathrm{TV}} \le \tfrac{1}{4} \Big\}.\footnotemark
\]
\footnotetext{The particular choice of $1/4$ is inessential: any fixed 
$\varepsilon \in (0,1)$ would give an equivalent definition of the mixing 
time up to constant factors.}
gives an equivalent definition of the mixing time up to constant factors.

\subsection{Main result}
We now turn to our main results concerning the Glauber dynamics of the
$p$-SOS model under entropic repulsion. While equilibrium aspects of
entropic repulsion are by now well understood, rigorous dynamical results
had previously been limited to the SOS case $p=1$
\cite{caputo2014dynamics}. Our aim here is to establish exponential lower
bounds on the mixing time, showing that entropic repulsion induces slow
mixing uniformly across the entire regime $1<p<\infty$.
However, our approach does not yield a sharp exponential lower bound of
the form $\exp(cL)$. This limitation originates from the behavior of the
one-point tail probabilities in the $p$-SOS model for $p>1$, which exhibit
abrupt variations between successive height levels. In particular, unlike
in the classical SOS case ($p=1$), the ratio between
$\hat\pi_p(\eta_0 \ge H(p,L))$ and
$\hat\pi_p(\eta_0 \ge H(p,L)-1)$ may vary by subpolynomial factors in $L$.
As a consequence, the scale of the critical bottleneck depends sensitively
on the system size $L$, and the present bottleneck argument cannot produce
a uniform lower bound of order $\exp(cL)$.

This difficulty appears to reflect a broader structural feature of the
$|\nabla\phi|^p$ models when $p>1$. Recent work of \cite{chen2025limiting} on
the geometry of level lines under entropic repulsion shows that, in contrast
to the SOS case, the one-point height distribution induces a separation of
scales between successive level lines, preventing the existence of a single
characteristic height scale governing large deviations. Their analysis,
which concerns equilibrium fluctuations rather than dynamics, likewise
relies on comparison and stochastic domination arguments instead of sharp
energetic estimates, highlighting the same lack of uniform control across
adjacent height levels that arises in the present work.

It therefore remains unclear whether this phenomenon reflects the true
behavior of the mixing time itself, or merely a limitation of the analytical
method employed here. In particular, it is an open problem whether more
refined techniques could upgrade the lower bound to a genuinely exponential
one of the form $\exp(cL)$, or whether the stretched-exponential scale
$\exp\!\bigl(cL^{1-o(1)}\bigr)$ is in fact optimal. We leave this question
for future work.

\begin{maintheorem}
\label{maintheorem}
Fix $1<p<\infty$. For sufficiently large inverse temperature $\beta$, 
there exists a constant $c:=c(\beta,p)>0$ such that the following holds.  
Consider the $(2+1)$-dimensional $p$-SOS model on $\Lambda_L$ with zero boundary condition, a floor at $0$, and a ceiling at $n_+$ satisfying $n_+\geq \log L$.  
For all sufficiently large $L \in \mathbb{N}$, the mixing time $\tmix$ of the associated Glauber dynamics satisfies
\begin{equation}
\label{eq-tmix-principale}
  \exp{\left(cL^{1-o(1)}\right)} \;\leq\; \tmix.
\end{equation}
\end{maintheorem}

Beyond the global lower bound on the mixing time established in Theorem~\ref{maintheorem}, one may ask how the dynamics actually approaches equilibrium in finer detail. In fact, starting from the flat configuration at height zero, the interface does not jump directly to its typical height $H(p,L)$, but rather moves through a cascade of intermediate metastable stages. At each stage the interface remains essentially flat, and upward transitions to the next level occur only through very rare large–deviation events, corresponding to the nucleation of a sufficiently large droplet at that height. This picture suggests that equilibration is not a smooth ascent, but rather a step–by–step mechanism in which droplets of increasing height appear successively and expand to cover the system.

Before presenting our second main result, let us emphasize the distinction 
with Theorem~\ref{maintheorem}. While the global bound there applies uniformly 
to all $p \in (1,\infty)$, the refined analysis below reveals a dependence on 
the specific value of $p$. In particular, the dynamics requires an 
exponentially long time to lift the interface from the floor up to 
intermediate levels $aH(p,L)$, and the precise exponent reflects the strength 
of fluctuations encoded by $p$. This behavior is naturally interpreted in terms 
of metastability: the growth proceeds layer by layer, and each additional step 
must overcome an entropic barrier whose cost scales with the system size.

To formalize this dependence, we introduce the notation
\begin{align}\label{d(p)}
    d(p):=\begin{cases}
p,& 1<p<2,\\[6pt]
2,& 2\le p<\infty.
\end{cases}
\end{align}

\begin{maintheorem}
\label{maintheorem2}
Fix $1<p<\infty$ and sufficiently large inverse temperature $\beta$. 
There exists $c:=c(\beta,p)>0$ such that the following holds. 
Consider the Glauber dynamics $(\eta(t))_{t\ge0}$ for the $p$-SOS model 
on $\Lambda=\{1,\dots,L\}^2$ with zero boundary condition, a floor at $0$, 
and a ceiling at $n_+$ with $n_+\geq\log L$. The dynamics starts from the 
all-zero configuration. For any $a \in (0,1)$, define
\[
\Omega_{p,a, L}=\Bigl\{\eta\in\Omega_{\Lambda_L,n_+}^+:\ 
\#\{x\in\Lambda:\eta_x\ge aH(p,L)\}\ \ge \frac{9}{10}|\Lambda|\Bigr\}.
\]
Define
\[
\tau_{p,a, L}=\min\{t\ge0:\eta(t)\in\Omega_{p,a, L}\}.
\]
Then $\lim_{L\rightarrow\infty}\pi^p_\Lambda(\Omega_{p,a, L})=1$, and \[
\lim_{L\rightarrow\infty}\mathbb{P}\!\left(\exp{\left(cL^{a^{d(p)}-o(1)}\right)}\le\tau_{p,a, L}\right)=1.
\]
\end{maintheorem}
\subsubsection{Metastable States Induced by Entropic Repulsion}
The metastable evolution of the interface under entropic repulsion can be naturally described in terms of the droplet picture. By a droplet we mean a finite region where the interface locally attains height $(h+1)$, embedded in a background at typical height $h$. In this sense, a droplet corresponds to a local phase $(h+1)$ appearing within the dominant phase $h$, and its energetic cost is determined by the length of the surrounding contour. Thus, the nucleation probability of a droplet is directly governed by the geometric properties of its contour. This perspective is consistent with the classical analysis of metastability in the two-dimensional Ising model through droplet nucleation \cite{schonmann1998wulff}, as well as with the SOS model dynamics studied in \cite{caputo2014dynamics}. In the present work, we extend this philosophy to the generalized $p$-SOS model. Inside the domain, the droplet interior typically remains close to the average height $aH(p,L)$, while near the boundary the interface gradually decreases due to the zero boundary condition. It is therefore more natural to regard the profile as smoothly interpolating between $aH(p,L)$ and $0$, rather than dropping abruptly. The energetic barrier is then concentrated along the droplet edge, where the interface must absorb a height difference of order $aH(p,L)$.

Since the $p$-SOS Hamiltonian contains gradient terms of the form $|\nabla \eta|^p$, the energy cost associated with such a boundary discrepancy scales as $(aH(p,L))^p$. For $p=1$, this cost is linear, making droplet nucleation comparatively easier, whereas for $p>1$ the cost grows much faster, significantly suppressing the transition probability. As a consequence, for the same fraction $a \in (0,1)$, the hitting time $\tau_{p,a,L}$ becomes substantially larger as $p$ increases. This explains, on an intuitive level, why Theorem~2 exhibits a $p$-dependent exponent $d(p)$ in the lower bound.

In summary, the metastable states induced by entropic repulsion arise from the combined effect of entropy gain and boundary rigidity. This mechanism, already established in the Ising model \cite{schonmann1998wulff} and in the classical SOS dynamics \cite{caputo2014dynamics}, persists in the $p$-SOS setting, where the rigidity effect is further amplified, leading to $p$-dependent time scales for mixing and equilibration.

\subsection{Challenges and Strategy} In two dimensions, both discrete and continuous spin systems exhibit rich interface phenomena. When the spin values are unbounded, the discrete gradient model is typified by the SOS model, while its continuous counterpart is captured by the Gaussian Free Field (GFF). Each of these models possesses a well-developed set of structural and probabilistic properties.  In what follows, we discuss the difficulties in generalizing from these well-understood frameworks and the conceptual obstacles encountered when transferring ideas from the SOS model and the GFF to the broader $p$-SOS context.

\subsubsection{Continuous vs. Discrete gradient models}
We examine the difference between discrete and continuous models from the perspective of gradient models. A prototypical example of a discrete gradient model is the integer-valued SOS model, while the Gaussian Free Field (GFF) represents the canonical case of a continuous gradient model.

In the SOS model, the height function takes integer values, and its structural properties make contour representations particularly effective. Analytical tools such as cylinders, level lines, and Peierls-type arguments are available, and these geometric approaches are well-suited to describe structural phenomena such as layering transitions, wetting, and entropic repulsion. In contrast, the GFF involves a real-valued height field and is well-analyzed through analytic tools such as Green functions, harmonic functions, and random walk representations. However, extending results from the SOS model to the generalized $p$-SOS model poses fundamental difficulties. For $p=1$, the contour representation is well-defined and powerful, but for $1<p<\infty$, the energy term becomes nonlinear, and such contour-based analyses are no longer directly applicable. Consequently, there exist significant technical barriers to generalizing the phenomena observed in the SOS model to the $p$-SOS setting.

\subsubsection{Technical Challenges in the $p$-SOS Model}
The $p$-SOS model is defined by a nonlinear energy term of the form
\[
\sum_{x \sim y} |\eta_x - \eta_y|^p,
\]
where $p > 1$. This nonlinearity leads to a fundamentally different mathematical structure compared to the classical SOS model ($p = 1$), introducing several technical challenges.\\
Many analytical techniques used in the SOS model—such as contour representation and Peierls-type arguments—rely on the subadditivity of the Hamiltonian and the ability to decompose height configurations into geometric objects. A key difficulty in extending Peierls-type arguments to the $p$-SOS model arises from the nonlinearity of the energy function. In the classical $p=1$ SOS model, many arguments rely on the ability to shift the heights inside a given contour by a fixed amount without affecting the energy. For example, for any non-negative integers $x \leq y$ and any $h > 0$, one has
\[
|y| - |x| = |y - x| = |y + h - (x + h)| = |y + h| - |x + h|,
\]
so the energy difference remains invariant under such shifts.\\
However, for $p > 1$, this invariance no longer holds due to the nonlinearity of the function $x \mapsto |x|^p$. A simple derivative argument shows that for $x \leq y$ and $h > 0$,
\[
|y + h|^p - |x + h|^p \geq p|x + h|^{p-1}|y - x|,
\]
which implies that the energy difference depends on both $x$ and the shift $h$. This lack of shift-invariance complicates the use of classical Peierls-type constructions, where contour-based mappings are frequently used to estimate large deviation probabilities.

As a consequence, while the one-point large deviation principle (LDP) for the SOS model can be established relatively easily (see Proposition 3.9 in \cite{caputo2014dynamics}), the extension to the $p$-SOS model requires more sophisticated analytical tools. This is evident in \cite{lubetzky2016harmonic}, where Theorems 3.1, 5.1, and 5.5 involve techniques from harmonic and $p$-harmonic analysis to overcome the nonlinear nature of the energy landscape in the $p$-SOS setting.
Consequently, extending some phenomena from SOS to the $p$-SOS model remains challenging. In particular, to estimate the lower bound on the mixing time under entropic repulsion, one needs a Peierls-type inequality adapted to this setting. We extended the inequality presented in Proposition 4.1 of \cite{lubetzky2016harmonic} to the general $p$-SOS model. This allowed us to establish that the mixing time under entropic repulsion exhibits an exponential lower bound.

\subsubsection{Strategy}
To prove the exponential lower bound on the mixing time, we analyze the time it takes for the dynamics to reach a “typical” region of configurations that dominates the equilibrium measure. The argument follows two main steps. In the first step, we identify a target set of configurations that, at equilibrium, carries almost the entire probability mass. Intuitively, these are the configurations where the interface has already risen close to its entropically favored height. Demonstrating that this set has near-unit probability ensures that reaching it is essentially equivalent to reaching equilibrium. In the second step, we design an initial distribution supported on atypical configurations, where the interface is still close to the floor and lacks any macroscopic upward fluctuations. Starting from such an initial state, the dynamics must overcome significant energetic and entropic barriers to enter the typical region. We then estimate the probability that the dynamics makes this transition within an exponential time window. This requires two ingredients: controlling the likelihood of leaving the atypical region too quickly, and bounding the probability of creating the critical upward fluctuations that drive the system toward equilibrium. Both steps are shown to have exponentially small probabilities. Finally, the analysis is first carried out in the simpler “floorless” setting, where technicalities are easier to manage, and then transferred back to the original model with a hard floor using standard comparison estimates. This strategy captures the core mechanism behind the slow mixing: the dynamics must wait for a rare fluctuation to initiate the transition to the typical height level, resulting in an exponentially large mixing time.

\subsubsection{Resolution}
We establish a general Peierls-type inequality for the $p$-SOS model in Theorem~\ref{prop4.1}, extending classical arguments from \cite{brandenberger1982decay}. This result generalizes the standard contour-based methods from the SOS case ($p = 1$) to the nonlinear setting $p > 1$. In particular, Theorem~\ref{prop4.1} includes a third term of the form $e^{-4\beta h}|\gamma|\log|\gamma|$, which, while slightly rougher than the bound presented in Proposition~4.1 of \cite{lubetzky2016harmonic}, has the advantage of being applicable uniformly across all $p$-SOS models. This term becomes negligible for sufficiently large $\beta$ and $h$, and in fact vanishes naturally throughout the proof without affecting the core argument. Thus, this extension provides a robust foundation for applying Peierls-type techniques in the $p$-SOS setting, playing a crucial role in the subsequent analysis of the exponential lower bound on the mixing time.\\
The second key component concerns the definition and role of the typical height $H(p, L)$ under entropic repulsion. In the SOS case $p=1$, Caputo et al.~\cite{caputo2014dynamics} 
showed that the typical height satisfies 
$H(1,L) = \tfrac{1}{4\beta}\log L + o(1)$. 
This was later sharpened in \cite{caputo2016scaling}, 
where it was proved that 
$H(1,L) = \lfloor (1/4\beta)\log L \rfloor$, 
thus providing a precise characterization without any $o(1)$ error term.\\
By contrast, for the $p$-SOS model, no such precise expression for $H(p, L)$ is available. Instead, \cite{lubetzky2016harmonic} provides only a rough estimate. In our approach, we adopt the probabilistic criterion used there, defining $H(p, L)$ as the smallest height such that
\[
\hat{\pi}_\Lambda^p(\eta_0 \geq H(p, L)) \geq \frac{5\beta}{L}.
\]
This definition allows us to construct a suitable initial configuration without requiring an explicit formula for $H(p, L)$, thereby enabling the derivation of a lower bound on the mixing time. Moreover, this strategy offers a practical and computationally efficient way to extend techniques developed for the SOS model to the more general $p$-SOS setting.

%\red{subsection Relatied Work 부분 추가하기 여기에서 fluctuation이랑 추가적으로 조금 더 찾아보기}

\subsection{Open Problems}
This work represents a first step toward understanding the behavior of Glauber dynamics under entropic repulsion in the $p$-SOS model, but several important questions remain open, since although we established a stretched-exponential lower bound on the mixing time of the form $\exp\!\bigl(cL^{1-o(1)}\bigr)$, it is not yet clear whether this bound is sharp; indeed, while in the classical SOS model \cite{lubetzky2016harmonic} proved a matching upper bound of order $\exp(O(L))$ by combining a refined canonical path argument with monotonicity, censoring techniques, and cluster expansion, for the $p$-SOS model it remains unclear whether a comparable upper bound can be derived or whether the slower stretched-exponential time scale suggested by our lower bound genuinely reflects the true behavior of the dynamics, a difficulty that stems from the lack of subadditivity and the increased complexity of the energy penalties associated with large gradients, which hinder a direct extension of the techniques available in the SOS case; in this context, although \cite{caputo2017entropic} developed a cluster expansion adapted to general $p$-SOS interfaces and provided tools that may prove useful for deriving upper bounds, whether these or other methods can lead to a sharper lower bound of order $\exp(cL)$ or instead to a matching upper bound confirming a stretched-exponential mixing time remains an open problem and an interesting direction for future research.

A further open question concerns the sharpness of the exponential scale. 
While both lower and upper bounds of the form $\exp(cL^{1-o(1)})$ or $\exp(cL^{a^{d(p)}-o(1)})$ 
have been established for the SOS model \cite{caputo2014dynamics} and related settings, 
the precise value of the constant $c$ has never been determined. 
It remains unclear whether $c$ can be explicitly identified, or at least characterized, 
and whether this constant depends monotonically on $p$ in the $p$-SOS model. 
Even understanding the qualitative behavior --- for instance, whether $c$ increases or decreases as $p$ varies --- 
would provide valuable insight into the dynamics under entropic repulsion.

Another natural direction is to extend the present results beyond the very low-temperature regime. While our analysis relies heavily on contour expansion techniques valid for large $\beta$, the behavior near the roughening transition $\beta_R$ remains largely mysterious. In particular, obtaining rigorous mixing time estimates for $p$-SOS dynamics in the regime $\beta$ just above $\beta_R$, where standard Peierls-type arguments break down, poses a significant challenge and requires new ideas beyond cluster expansion.

\subsection{Organization}
The remainder of the paper is organized as follows. In Section~\ref{sec:notation}, we introduce the $p$-SOS model, the Glauber dynamics, and the main probabilistic tools, including the contour representation, the DLR formalism, and the FKG inequality. Section~\ref{sec:equilibrium} presents the equilibrium properties of the model under entropic repulsion, such as the large deviation estimates and the typical height $H(p,L)$. In Section~\ref{sec:maintheorem}, we prove Theorems~\ref{maintheorem} and~\ref{maintheorem2}, establishing exponential lower bounds on the mixing time for the $p$-SOS model. Finally, Section~\ref{appendix} collects several technical computations that would otherwise interrupt the flow of the main text. In particular, detailed case-by-case analyses for $1<p<2$, $p=2$, and $2<p<\infty$ are deferred to the Appendix.
\begin{comment}
    \subsection{Acknowledgements} The author thanks to Professor Kyeongsik Nam for his guidance and  support throughout the development of this work.
\end{comment}

\section{Model and Technical Framework}\label{sec:notation}
In this section we introduce the technical framework and probabilistic tools that will be used throughout the paper. We first recall the contour representation of the $p$-SOS model, which provides a convenient geometric description of height configurations and plays a central role in our large deviation and Peierls-type arguments. We then review the Dobrushin--Lanford--Ruelle (DLR) formalism for Gibbs measures, which allows us to formulate both finite and infinite-volume settings in a consistent way. Finally, we present an application of the FKG inequality, which will be repeatedly used to compare probabilities of monotone events under different boundary conditions.

\subsection{Contour representation}

We now introduce the notion of contours for the $p$-SOS model. Contours encode the level lines of the height function and allow us to describe deviations of the interface from its typical height in terms of geometric objects. In particular, they are the key tool in our large deviation estimates and in the analysis of metastable 
transitions, where the nucleation of critical droplets can be expressed through the formation of suitable contours.

\begin{definition}[Contour]
We denote by $\mathbb{Z}^{2*}$ the dual lattice of $\mathbb{Z}^2$, whose sites are placed at the centers of the unit squares of the original (primal) lattice. A \emph{bond} in the dual lattice refers to a line segment connecting two neighboring sites in $\mathbb{Z}^{2*}$. Given two sites $x, y \in \mathbb{Z}^2$, we say that they are \emph{separated} by a dual bond $e$ if the Euclidean distance from both $x$ and $y$ to $e$ is exactly $1/2$. A pair of orthogonal dual bonds meeting at a vertex $x^* \in \mathbb{Z}^{2*}$ is called \emph{linked} if both bonds lie on the same side of one of the two diagonal lines of slope $\pm 1$ passing through $x^*$.
\begin{enumerate}
    \item The bonds are distinct except for $e_0 = e_n$, forming a closed loop.
    \item Each pair of consecutive bonds $e_i, e_{i+1}$ shares a common vertex in the dual lattice $\mathbb{Z}^{2*}$.
    \item If four bonds $e_i, e_{i+1}, e_j, e_{j+1}$ share a vertex $x^* \in \mathbb{Z}^{2*}$, then $(e_i, e_{i+1})$ and $(e_j, e_{j+1})$ must be \emph{linked}, i.e., lie on the same side of one of the two $45^\circ$ diagonals through $x^*$ in the dual lattice.
\end{enumerate}
We denote by $|\gamma|$ the length of $\gamma$, and by $\Lambda_\gamma$ the set of sites in $\mathbb{Z}^2$ that are strictly enclosed by $\gamma$. We write $|\Lambda_\gamma|$ for the number of sites contained in the interior of $\gamma$.\\
The boundary associated with $\gamma$ is defined as follows:
\begin{itemize}
    \item $\partial\gamma$: the set of sites in $\mathbb{Z}^2$ that are adjacent to $\gamma$, either via Euclidean distance $1/2$ in the primal lattice from a dual bond in $\gamma$, or via distance $1/\sqrt{2}$ from a vertex in $\mathbb{Z}^{2*}$ where two non-linked bonds of $\gamma$ meet.
    \item $\partial^+_\gamma := \partial\gamma \cap \Lambda_\gamma$ and $\partial^-_\gamma := \partial\gamma \setminus \partial^+_\gamma$.
    \item $\eta|_{\partial_\gamma^+} \ge j 
\quad \text{to mean that } \eta_x \ge j \;\; \text{for all } x \in \partial_\gamma^+$
    \item $\eta|_{\partial_\gamma^-} \le j-1 
\quad \text{to mean that } \eta_y \le j-1 \;\; \text{for all } y \in \partial_\gamma^-.$
\end{itemize}
\end{definition}

\begin{definition}[$h$-contour]\label{hcontour}
Given a configuration $\eta$, a contour $\gamma$ is said to be an $h$-contour if
\[
\eta|_{\partial_\gamma^+} \ge h \quad \text{ and }\quad \eta|_{\partial_\gamma^-} \le h-1.
\]
This represents a level drop across $\gamma$, from height at least $h$ inside to at most $h-1$ outside.

Let $\mathcal{C}_{\gamma,h}$ denote the event that $\gamma$ is an $h$-contour in configuration $\eta$, i.e.,
\[
\mathcal{C}_{\gamma,h} := \left\{ \eta : \eta|_{\partial_\gamma^+} \ge h \quad \text{ and }\quad \eta|_{\partial_\gamma^-} \le h-1\right\}.
\]

\end{definition}

\begin{remark}
For example, the elementary square contour of length $4$ surrounding a site $x \in \mathbb{Z}^2$ is an $h$-contour if $\eta(x) \ge h$ and all its neighbors are at height at most $h-1$. This situation can also occur for multiple levels: if the inside height is at least $\max(h,h')$ and the outside height is at most $\min(h,h')-1$, then the same geometric contour can simultaneously be an $h$-contour and an $h'$-contour for $h \ne h'$].
\end{remark}

\subsection{DLR Formalism and Gibbs Measures}\label{sub:dlr}

Recall $\Lambda \subset \mathbb{Z}^2$ be a finite subset and denote by $\Omega_\Lambda = \mathbb{Z}^\Lambda$ the space of integer-valued configurations on $\Lambda$. A \emph{Gibbs specification} is a family of probability kernels 
$\{\gamma_\Lambda^\phi\}_{\Lambda \subset \mathbb{Z}^2, \, \phi:\partial\Lambda\rightarrow\mathbb{Z}}$, 
where each $\gamma_\Lambda^\phi$ is defined via the Boltzmann–Gibbs distribution 
associated with the Hamiltonian of the model under consideration.

We consider the configuration space $\Omega = \mathbb{Z}^{\mathbb{Z}^2}$ equipped with the product $\sigma$-algebra $\mathcal{F}$ generated by local cylinder events. The set of probability measures on $(\Omega, \mathcal{F})$ is denoted by $\mathcal{M}_1(\Omega)$. A function $f: \Omega \rightarrow \mathbb{R}$ is said to be local if it depends only on finitely many coordinates, i.e., there exists $\Lambda \subset \mathbb{Z}^2$ such that $f(\eta) = f(\eta')$ whenever $\eta|_\Lambda = \eta'|_\Lambda$.

\begin{definition}[DLR Condition]
Let $\{\gamma_\Lambda^\phi\}$ be a Gibbs specification indexed by finite 
$\Lambda \subset \mathbb{Z}^2$ and boundary conditions $\phi:\partial\Lambda \rightarrow \mathbb{Z}$.  
A probability measure $\mu \in \mathcal{M}_1(\Omega)$ is said to satisfy the 
Dobrushin–Lanford–Ruelle (DLR) condition if, for every finite $\Lambda \subset \mathbb{Z}^2$ 
and every bounded measurable function $f:\Omega \rightarrow \mathbb{R}$,
\[
\mu(f) = \int \gamma_\Lambda^{\eta|_{\partial\Lambda}}(f)\, \mu(d\eta).
\]
\end{definition}
This relation characterizes the consistency between the global measure $\mu$ 
and the family of local Gibbs specifications $\gamma_\Lambda^\phi$. 
A probability measure that satisfies the DLR condition is referred to as an 
\emph{infinite-volume Gibbs measure}. Importantly, such measures encode the local 
equilibrium structure of the system and allow a rigorous formulation of phase 
coexistence and correlation decay.

In the context of our study, we denote by $\hat{\pi}^{p}$ the infinite-volume Gibbs 
measure for the $p$-SOS model without floor or ceiling, and by $\bar{\pi}^{p}$ the 
corresponding measure under a hard floor at height zero. Both measures satisfy the 
DLR condition with respect to the appropriate finite-volume specifications. As we 
shall see in later sections, this structure is essential in justifying decompositions 
of probability expressions and applying recursive arguments based on domain restriction. For further background on the DLR formalism, we refer the reader to the standard 
reference by Friedli and Velenik~\cite{friedli2018statistical}.

\subsection{FKG Inequality}\label{sub:fkg}

The FKG inequality (named after Fortuin, Kasteleyn, and Ginibre) is a fundamental correlation inequality in the theory of lattice spin systems, first established in \cite{fortuin1971correlation}. It asserts that, under suitable conditions, increasing events are positively correlated. A Gibbs measure is said to have \emph{attractive interactions} (or to be \emph{ferromagnetic}) if the coupling constants are nonnegative, so that neighboring variables tend to align. This monotonicity is precisely the condition under which the FKG inequality holds. 
Formally, a partial order can be defined on the configuration space $\Omega$ by pointwise comparison: for $\eta, \eta' \in \Omega$, we write $\eta \leq \eta'$ if and only if $\eta_x \leq \eta'_x$ for all $x \in \mathbb{Z}^2$. An event $E \subset \Omega$ is \emph{increasing} if $\eta \in E$ and $\eta \leq \eta'$ imply $\eta' \in E$. Similarly, a function $f:\Omega \rightarrow \mathbb{R}$ is \emph{nondecreasing} if $f(\eta) \leq f(\eta')$ whenever $\eta \leq \eta'$. For a Gibbs measure $\mu$ with attractive interactions, the FKG inequality states 
that for any pair of bounded nondecreasing local functions $f$ and $g$,
\[
\mu(fg) \;\geq\; \mu(f)\,\mu(g),
\]
where $(fg)(\eta) := f(\eta)\,g(\eta)$ denotes the pointwise product. 
In particular, if $E$ and $F$ are increasing events, then 
$f=\mathbf{1}_E$ and $g=\mathbf{1}_F$ yield 
$\mu(fg) = \mu(E \cap F) \ge \mu(E)\mu(F)$.
In particular, if $E$ and $F$ are increasing events, their indicator functions $\mathbf{1}_E$ and $\mathbf{1}_F$ are nondecreasing, implying that the occurrence of one increases the likelihood of the other. A key feature of the FKG inequality is that it applies under arbitrary boundary conditions and general external fields, provided the interaction is attractive. Equivalently, in the setting of finite-volume Gibbs measures with boundary condition $\phi$ on a finite domain $\Lambda \subset \mathbb{Z}^2$, one has
\[
\mu^\phi_\Lambda(fg) \;\geq\; \mu^\phi_\Lambda(f)\,\mu^\phi_\Lambda(g)
\]
for all bounded nondecreasing local functions $f$ and $g$. In our context, the FKG inequality will be used to compare probabilities under different boundary conditions and to deduce product-type lower bounds for certain increasing events. This property is particularly useful in the analysis of contour events, where the interior of a contour is influenced by its boundary values.

\subsection{Exponential Decay of Correlations}\label{exponentiallydecay}

For sufficiently large $\beta$, two-dimensional interface models such as the SOS, DG, and more generally the $p$-SOS model, exhibit exponential decay of truncated correlations between local observables. 
This was rigorously established in the classical work of 
Brandenberger and Wayne~\cite{brandenberger1982decay} via a convergent 
low-temperature cluster expansion. 

Formally, let $f(\eta_{\Lambda_1})$ and $g(\eta_{\Lambda_2})$ 
be bounded local functions depending only on the disjoint finite sets 
$\Lambda_1,\Lambda_2 \subset \mathbb{Z}^2$. 
Then there exist constants $C_1,C_2>0$ (depending only on $\beta$ and $p$) such that
\[
\big|\,\hat{\pi}^p(f g)
       - \hat{\pi}^p(f)\,\hat{\pi}^p(g)\,\big|
   \;\le\; C_1 \exp\!\big(-C_2\,\mathrm{dist}(\Lambda_1,\Lambda_2)\big),
\]
where $\hat{\pi}^p$ denotes the infinite-volume Gibbs measure of the $p$-SOS model.
Equivalently, truncated correlations of local observables decay exponentially 
in the distance between their supports.

\section{Equilibrium Properties under Entropic Repulsion}\label{sec:equilibrium}

The following result parallels Theorem 3.1 of \cite{caputo2014dynamics} for the SOS model ($p=1$) and Theorem 2 and Thm 4 of \cite{lubetzky2016harmonic} for the Discrete Gaussian model ($1<p<\infty$). In the presence of entropic repulsion, nearly all heights in $\Lambda_L$ concentrate around the typical height $H(p,L)$. Relying on the mechanisms identified in these works, we formulate the corresponding statement for the $p$-SOS model.
\begin{theorem}
\label{thm-h-event}
Fix $p \ge 1$. For any sufficiently large $\beta$, let $\varepsilon=\varepsilon(\beta)>0$ be such that $\varepsilon(\beta)\rightarrow0$ as $\beta\rightarrow\infty$. Then there exists $L_0=L_0(\beta,p)$ such that, for all $L\ge L_0$ and any box $\Lambda\subset\mathbb{Z}^2$ of side length $L$,
\[
\bar{\pi}^p_\Lambda\!\left(\,\#\{v\in\Lambda:\eta_v\ge H(p,L)-1\}\,\ge\,(1-\varepsilon(\beta))|\Lambda|\,\right)
\longrightarrow 1
\quad\text{as }L\rightarrow\infty.
\]
\end{theorem}

For the mixing-time analysis, we shall work with a fixed ceiling height $n_+$. 
However, the choice of ceiling has little effect on the stationary measure: the difference between the ceiling-constrained measure $\pi^p_\Lambda$ (floor at $0$, ceiling at $n_+$) and the floor-only measure $\bar{\pi}^{p, 0}_\Lambda$ (no ceiling) is asymptotically negligible. 
This allows us to transfer estimates from $\bar{\pi}^{p, 0}_\Lambda$ to $\pi^{p, 0}_\Lambda$ at negligible cost. 
The next statement is the $p$-SOS analogue of Lemma~3.2 in~\cite{caputo2014dynamics}.

% Lemma : comparisonwithnoandyes

\begin{proposition}
\label{comparisonwithnoandyes}
Let $p \ge 1$ and $\beta \ge 1$ be sufficiently large. Then, for any event 
$A \subset \{0, \dots, n_+\}^\Lambda$, there exists a constant $c_1>0$ 
(as introduced in Proposition~\ref{LDP_domain}) such that
\[
\bar{\pi}_\Lambda^{p,0}(A)
\le 
\pi_\Lambda^{p,0}(A)
\le 
\left[
1 + 2\left(c_1L \log L\, e^{-2\beta n_+}+ L^2\, \hat{\pi}^{p}(\eta_0 \ge n_+/2)
+ L^{-4\beta + 2}\right)\right]
\bar{\pi}_\Lambda^{p,0}(A).
\]
\end{proposition}
This inequality compares the full Gibbs measure $\pi_\Lambda^{p,0}$, defined with both floor and ceiling constraints, to the unconstrained Gibbs measure $\hat{\pi}_\Lambda^{p,0}$. The error terms arise from bounding the effect of the ceiling and the finite domain, and are small provided $\beta$ and $n_+$ are large. The contour length term $c_1L \log L$ corresponds to the outer boundary of the square domain $\Lambda$.

The following result provides a Peierls-type upper bound on the probability of the occurrence of a given $h$-contour in the $p$-SOS model with a floor at zero. For the SOS case ($p=1$), an analogous statement was established in Proposition~3.6 of~\cite{caputo2014dynamics}, while for the Discrete Gaussian ($p=2$), the corresponding estimate appears as Proposition~4.1 in~\cite{lubetzky2016harmonic}. Here, we extend these bounds to the general $p$-SOS setting, combining Peierls estimates with the entropic repulsion effect to obtain a bound valid for all $p\geq 1$.
% Prop 4.1 사용
\begin{theorem}\label{prop4.1}
Consider the $p$-SOS model($p\geq1$) in a finite connected subset $\Lambda\subset\mathbb{Z}^2$ with a floor at $0$ and zero boundary condition on $\partial\Lambda$. 
Let $\gamma\subset\Lambda$ be a fixed closed dual circuit. For sufficiently large $\beta$, there exists a constant $c_1 > 0$ such that
\[
\bar{\pi}^{p, 0}_\Lambda(\mathcal{C}_{\gamma,h})
\;\le\;
\exp\!\left(
  -\,\beta\,|\gamma|
  \;+\;\hat{\pi}^p(\eta_0\ge h)|\Lambda_\gamma|
  \;+c_1\;e^{-4\beta h}\,|\gamma|\,\log|\gamma|
\right).
\]
\end{theorem}
To prove Theorem~\ref{prop4.1}, we first establish a Peierls-type inequality for the $p$-SOS model without ceiling constraint. While the inequality is structurally similar to the one known for the standard SOS model, the nonlinear nature of the $p$-norm interaction requires certain technical modifications in the argument, which we present for completeness.

The next lemma provides the basic Peierls bound in the absence of both floor and ceiling, for the $p$-SOS model with arbitrary boundary conditions. This estimate will serve as a fundamental ingredient in the proofs of the main results stated above. For the SOS case ($p=1$), the same bound appears in Lemma~3.7 of~\cite{caputo2014dynamics}, and here we extend it to the general $p$-SOS setting ($p \ge 1$) without any change in the form of the estimate.
\begin{lemma}\label{peierlsimple}
Let $\Lambda\subset\mathbb{Z}^2$ be finite and let $p\ge 1$. 
For any contour $\gamma\subset\Lambda$ and any $h\in\mathbb{Z}$,
\begin{equation}
\label{lemma3.7for}
    \hat{\pi}^{p,0}_\Lambda\!\left(\mathcal{C}_{\gamma,h}\right) \;\le\;
    \exp\!\bigl(-\beta\,|\gamma| \bigr).
\end{equation}
Moreover, if $h'<h$ and $\gamma,\gamma'$ are contours with
$\Lambda_\gamma\subseteq\Lambda_{\gamma'}$, then
\begin{equation}
\label{eq2:lemma3.7}
 \hat{\pi}^{p,0}_\Lambda\!\left(
C_{\gamma,h} \,\middle|\, C_{\gamma',h'} \right) \;\le\; \exp\!\bigl(-\beta\,|\gamma| \bigr).
\end{equation}
\end{lemma}

\begin{proof}[Proof of Lemma~\ref{peierlsimple}]\label{peierlsimplepf}
The argument is identical to Lemma 3.7 of \cite{caputo2014dynamics}. We include it only to record the simple convexity observation needed for the extension to general $p\ge1$. We define the shift map $T_{\gamma}:\mathbb{Z}^\Lambda \rightarrow \mathbb{Z}^\Lambda$ by
\begin{equation}\label{mapping}
T_{\gamma}(\eta)_v =
\begin{cases}
\eta_v - 1, & \text{if } v \in \Lambda_\gamma, \\
\eta_v, & \text{otherwise}.
\end{cases}
\end{equation}
For any configuration $\eta \in \mathcal{C}_{\gamma,h}$, and any edge $e=\{x,y\}$ crossing $\gamma$ (oriented such that $x \in \Lambda_\gamma$ and $y \notin \Lambda_\gamma$), the $h$-contour condition ensures $\Delta_e := \eta_x - \eta_y \ge 1$. The key difference from the SOS model lies in the hamiltonian estimate for $p \ge 1$. Since $f(t)=|t|^p$ is convex, the gradient satisfies:$$|(T\eta)_x - (T\eta)_y|^p = (\Delta_e - 1)^p \le \Delta_e^p - 1 = |\eta_x - \eta_y|^p - 1.$$For edges $e \notin E_\gamma$, the gradient remains unchanged. Summing over all edges and noting $|E_\gamma| = |\gamma|$, we obtain the Hamiltonian bound:
\begin{equation}\label{eq:energy-drop}
H_p(T\eta) \le H_p(\eta) - |E_\gamma|
= H_p(\eta) - |\gamma|.
\end{equation}
This leads to $\hat{\pi}^{p,0}_\Lambda(T\eta) \ge e^{\beta|\gamma|}\hat{\pi}^{p,0}_\Lambda(\eta)$. The remainder of the proof, including the extension to the conditional estimate \eqref{eq2:lemma3.7}, proceeds via the same steps as those in \cite{caputo2014dynamics}.
\end{proof}
\begin{comment}
    Unlike the classical SOS model, where the linear structure of the energy 
allows rather explicit control of marginal distributions, the $p$-SOS model 
is more sensitive to boundary effects due to the nonlinear dependence of the 
energy on gradients. In particular, computing one-point marginals near the 
boundary becomes technically delicate. Nevertheless, this issue has been addressed in \cite{chen2025limiting}, where Lemma 2.8 and Lemma 7.10 establish the necessary estimates by combining the aforementioned tools with the exponential decay properties. We shall utilize these established bounds, as the associated error terms are exponentially small in the system size and thus negligible for our purposes. This provides a sufficient basis to apply the large deviation principle (LDP) for the $p$-SOS model near the boundary of a finite domain.
\end{comment}

Proposition \ref{LDP_domain} plays a pivotal role in the proofs of Theorem \ref{maintheorem} and Theorem \ref{maintheorem2}, as it provides an upper bound on the one-point large deviation probability for any finite domain $\Lambda$ and any parameter $1 < p < \infty$. Unlike sharp estimates specifically tailored for the infinite-volume Gibbs measure, this result offers a robust, albeit less sharp, bound for the event that the interface height at an arbitrary site $x \in \Lambda$ exceeds $h$. The historical context of such estimates varies significantly by the value of $p$. In the classical SOS model ($p=1$), both upper and lower bounds were established with considerable precision for arbitrary $x \in \Lambda$ in Proposition 3.9 of \cite{caputo2014dynamics}. For the Discrete Gaussian model ($p=2$), similar considerations were mentioned in the proof of Proposition 4.1 in \cite{lubetzky2016harmonic}, though explicit discussion for general finite domains remained limited.\\
Recently, more refined large deviation principles have been established in \cite[Lemma 2.8]{chen2025limiting} for $p=2$ and in \cite[Lemma 7.10]{chen2025limiting} for $p \in (1, \infty) \setminus \{2\}$. However, these sharper estimates typically require the finite domain $\Lambda$ to contain a ball $B_R(x)$ of radius $R$ depending on $h$, and in the case of $p \in(1, \infty)-\{2\}$, the height $h$ must be sufficiently large relative to the inverse temperature $\beta$. While such assumptions are justifiable in equilibrium under entropic repulsion, they pose a challenge in the study of mixing times. Analyzing the dynamics requires estimating the probability of bottleneck events that drive the system toward equilibrium, which necessitates bounds that are valid across arbitrary domains.\\
Therefore, we seek to establish a versatile estimate applicable to any finite domain $\Lambda$, all $h \geq 0$, and any $p \geq 1$. The proof is structured into two regimes. For sites $x$ located near the boundary(boundary region), we extend the Peierls-type argument from Lemma \ref{peierlsimple} and the combinatorial methods found in the proof of Proposition 3.9 of \cite{caputo2014dynamics}. For sites in the bulk(bulk region), we rely on the exponential decay of correlations established in \cite[Theorem 2]{brandenberger1982decay}. This approach was instrumental in the analysis by \cite{lubetzky2016harmonic} and is summarized in Subsection \ref{exponentiallydecay}.

\begin{proposition}\label{LDP_domain}
Fix $p\ge 1$ and for sufficiently large $\beta \gg 1$. Let $\Lambda\subset\mathbb{Z}^2$ be finite, $v\in\Lambda$, and $h\ge 0$. Then there exists a constant $c_1>0$ such that
\[
\hat{\pi}^{p,0}_\Lambda(\eta_v \ge h) \;\le\;
\begin{cases}\label{eq:LDP_domain}
c_1\,e^{-4\beta h}, & \text{if }\ \mathrm{dist}(v,\partial\Lambda)\le \log|\partial\Lambda|,\\[4pt]
\hat{\pi}^{p}(\eta_0 \ge h) + |\Lambda|^{-4}, & \text{if }\ \mathrm{dist}(v,\partial\Lambda)\ge \log|\partial\Lambda|.
\end{cases}
\]
In particular, for every $v$ within distance $\log|\Lambda|$ of the boundary,
\(
\hat{\pi}^{p,0}_\Lambda(\eta_v\ge h)\le c\,e^{-4\beta h}.
\)
\end{proposition}

% Pf : LDP_domain 사용
\begin{proof}[Proof of Proposition~\ref{LDP_domain}]\label{LDP_domainpf}
To obtain an upper bound on $\hat{\pi}_{\Lambda}^{p,0}(\eta_v \ge h)$, we analyze the behavior by partitioning the domain into a boundary layer and a bulk region, depending on the position of vertex $v$.

\noindent\textbf{Case 1. Boundary region ($\text{dist}(v, \partial\Lambda) \le \log |\Lambda|$).} In this regime, we apply a Peierls-type argument using a sequence of $h$ nested $i$-contours $(\gamma_1, \dots, \gamma_h)$ surrounding $v$. Since the weight of each contour is suppressed by $\hat{\pi}_{\Lambda}^{p,0}(\mathcal{C}_{\gamma, h}) \le \exp(-\beta|\gamma|)$ according to Lemma~\ref{peierlsimple}, the summation over these nested sequences follows the standard combinatorial framework for the $p=1$ SOS model, as established in \cite[Proposition 3.9]{caputo2014dynamics}. By classifying the contours according to their area and summing the corresponding exponential weights, we obtain the following uniform exponential decay upper bound:

\noindent\textbf{Case 2. Bulk region ($\text{dist}(v, \partial\Lambda) \ge \log |\Lambda|$).} When $v$ is sufficiently far from the boundary, we utilize the exponential decay of correlations that holds for the low-temperature $p$-SOS model. Since the distance between the local observable $A_{\{v\}} = \mathbf{1}_{\{\eta_v \ge h\}}$ and the boundary configuration indicator $B_{\partial\Lambda}$ is at least $\log |\Lambda|$, the results of \cite[Theorem 2]{brandenberger1982decay} imply the following:
$$|\hat{\pi}_{\Lambda}^{p,0}(\eta_v \ge h) - \hat{\pi}^p(\eta_0 \ge h)| \le Ce^{-c(\beta)\log |\Lambda|} = C|\Lambda|^{-c} \le |\Lambda|^{-4},$$where the term can be bounded by $|\Lambda|^{-4}$ for sufficiently large $\beta$.\\
Combining these two regimes, we conclude that when $v$ is near the boundary, the probability is controlled by the uniform Peierls bound $c_1 e^{-4\beta h}$. Conversely, when $v$ is in the bulk, the measure converges to the infinite-volume Gibbs measure. This yields the two-branch estimate for Proposition ~\ref{LDP_domain} as claimed.
\end{proof}

The argument combines the tools introduced in Section~\ref{sub:dlr} (DLR formalism) and Section~\ref{sub:fkg} (FKG inequality), which will be used repeatedly to control correlations and to decompose probabilities with respect to appropriate boundary conditions.

\begin{proof}[Proof of Theorem~\ref{prop4.1}]
\label{prop4.1pf}
By Proposition~\ref{comparisonwithnoandyes}, it suffices to work under
the floor measure $\bar{\pi}$ without a ceiling.
We begin with the estimate for $\mathcal{C}_{\gamma,h}$.
By ~\eqref{eq:energy-drop}, the shift map $T_\gamma$
produces an energy decrease of at least $|\gamma|$,
and hence
\[
e^{-\beta H_p(\eta)}
\le
e^{-\beta|\gamma|}\,e^{-\beta H_p(T_\gamma\eta)}.
\]
Since $T_\gamma$ is injective on $\{\eta\in \mathcal{C}_{\gamma,h}:\eta|_{\Lambda_\gamma}>0\}$, it follows that
\begin{equation}\label{eq:peierls-num}
\bar{\pi}^{p,0}_\Lambda\bigl(\mathcal{C}_{\gamma,h}\cap\{\eta|_{\Lambda_\gamma}>0\}\bigr)\;\le\; e^{-\beta|\gamma|}.
\end{equation}
To convert this into a bound under the floor measure, recall that
\begin{equation}
    \bar{\pi}^{p,0}_\Lambda(A)=\hat{\pi}^{p,0}_\Lambda(A\cap\{\eta\ge 0\})\,/\,\hat{\pi}^{p,0}_\Lambda(\eta\ge 0).    
\end{equation}
Conditioning on the configuration on $\partial^\pm_\gamma$ and using FKG/monotonicity, we obtain
\[
\bar{\pi}^{p,0}_\Lambda\!\bigl(\eta\in \mathcal{C}_{\gamma,h},\,\eta|_{\Lambda_\gamma}>0\bigr)
= \hat{\pi}^{p,0}_\Lambda\!\Bigl(
\mathbf{1}_{\{\eta^-_\gamma\le h-1\}}\,
\mathbf{1}_{\{\eta|_{\Lambda \setminus \Lambda_\gamma}\ge 0\}}\,
\mathbf{1}_{\{\eta^+_\gamma\ge h\}}\,
\hat{\pi}^{p,\eta^+_\gamma}_{\Lambda_\gamma}\!\bigl(\eta|_{\Lambda_\gamma}>0\bigr)
\Bigr)\Big/\hat{\pi}^{p,0}_\Lambda(\eta\ge 0).
\]
By monotonicity of the increasing event $\{\eta|_{\Lambda_\gamma}>0\}$ in the outer boundary condition,
\[
\hat{\pi}^{p,\eta^+_\gamma}_{\Lambda_\gamma}\bigl(\eta|_{\Lambda_\gamma}>0\bigr)
\;\ge\;
\hat{\pi}^{p,h}_{\Lambda_\gamma}\bigl(\eta|_{\Lambda_\gamma}>0\bigr)
\;=\;
\hat{\pi}^{p,0}_{\Lambda_\gamma}\bigl(\eta|_{\Lambda_\gamma}>-h\bigr).
\]
Applying FKG to the increasing events $\{\eta_x<h\}$ over $x\in \Lambda_\gamma$ yields
\[
\hat{\pi}^{p,0}_{\Lambda_\gamma}\bigl(\eta|_{\Lambda_\gamma}>-h\bigr)
\;\ge\;
\prod_{x\in \Lambda_\gamma}\Bigl(1-\hat{\pi}^{p,0}_{\Lambda_\gamma}(\eta_x\ge h)\Bigr).
\]
Consequently,
\begin{equation}\label{eq:pos-prob-factor}
\pi^{p,0}_\Lambda\!\bigl(\eta\in \mathcal{C}_{\gamma,h},\,\eta|_{\Lambda_\gamma}>0\bigr)
\;\ge\;
\pi^{p,0}_\Lambda(\mathcal{C}_{\gamma,h})\,
\prod_{x\in \Lambda_\gamma^{\mathrm{in}}}\Bigl(1-\hat{\pi}^{p,0}_{\Lambda_\gamma}(\eta_x\ge h)\Bigr).
\end{equation}
We now bound the product in \eqref{eq:pos-prob-factor}. By the one–point bounds from Proposition~\ref{LDP_domain} (applied inside $\Lambda_\gamma^{\mathrm{in}}$), for sufficiently large $\beta$,
\[
\hat{\pi}^{p,0}_{\Lambda_\gamma}(\eta_x\ge h)\;\le\;
\begin{cases}
c_1\,e^{-4\beta h}, & \text{if }\ \mathrm{dist}(x,\gamma)\le \log|\gamma|,\\[3pt]
\hat{\pi}^p(\eta_0\ge h)+|\Lambda_\gamma|^{-2}, & \text{otherwise}.
\end{cases}
\]
Thus,
\begin{equation}\label{eq:pos-prob-factor1}
    \prod_{x\in \Lambda_\gamma}\Bigl(1-\hat{\pi}^{p,0}_{\Lambda_\gamma}(\eta_x\ge h)\Bigr)
\;\ge\;
\exp\;\!\Bigl(- c_1\,e^{-4\beta h}\,|\gamma|\,\log|\gamma| \Bigr)\,
\exp\;\!\Bigl(- \hat{\pi}^p(\eta_0\ge h)\,A(\gamma)\Bigr).
\end{equation}
Combining this with \eqref{eq:peierls-num} and \eqref{eq:pos-prob-factor} , we obtain
\[
\pi^{p,0}_\Lambda(\mathcal{C}_{\gamma,h})
\;\le\;
\exp\;\!\Bigl(
-\beta|\gamma|
+ \hat{\pi}^p(\eta_0\ge h)\,A(\gamma)
+c_1 e^{-4\beta h}\,|\gamma|\,\log|\gamma|
\Bigr),
\]
which is the desired bound.
\end{proof}

We are now ready to prove Proposition~\ref{comparisonwithnoandyes} by combining the results established so far. Throughout this subsection we fix the box $\Lambda_L=\{1,\dots,L\}^2$. We write $\pi^{p,0}_{\Lambda_L}$ for the $p$-SOS measure with a floor at $0$ and a ceiling at $n_+$, and $\bar{\pi}^{p,0}_{\Lambda_L}$ for the floor-only measure (no ceiling).
\begin{proof}[Proof of Proposition~\ref{comparisonwithnoandyes}]
Let $Z^{p,0}_{\Lambda_L}$ and $\bar{Z}^{p,0}_{\Lambda_L}$ be their partition functions, respectively.
We have the identity

\[
\pi^{p,0}_{\Lambda_L}(A)
= \bar{\pi}^{p,0}_{\Lambda_L}(A)\,  \frac{\bar{Z}^{p,0}_{\Lambda_L}}{Z^{p,0}_{\Lambda_L}}.
\]
Since $\bar{Z}^{p,0}_{\Lambda_L}\ge Z^{p,0}_{\Lambda_L}$, it follows that
\[
\bar{\pi}^{p,0}_{\Lambda_L}(A)\;\le\;\pi^{p,0}_{\Lambda_L}(A).
\]
For the reverse inequality, note that
\[
\frac{Z^{p,0}_{\Lambda_L}}{\bar{Z}^{p,0}_{\Lambda_L}}
= \bar{\pi}^{p,0}_{\Lambda_L}(\eta\le n_+),
\qquad\text{hence}\qquad
\pi^{p,0}_{\Lambda_L}(A)
= \frac{\bar{\pi}^{p,0}_{\Lambda_L}(A)}{\bar{\pi}^{p,0}_{\Lambda_L}(\eta\le n_+)}
= \frac{\bar{\pi}^{p,0}_{\Lambda_L}(A)}{1-\bar{\pi}^{p,0}_{\Lambda_L}\!\left(\bigcup_{v\in\Lambda_L}\{\eta_v>n_+\}\right)}.
\]
By a union bound,
\begin{equation}\label{eq:union}
\pi^{p,0}_{\Lambda_L}(A)
\;\le\;
\frac{\bar{\pi}^{p,0}_{\Lambda_L}(A)}{1-\sum_{v\in\Lambda_L}\bar{\pi}^{p,0}_{\Lambda_L}(\eta_v>n_+)}.
\end{equation}
We now estimate the sum of ceiling-exceedance probabilities under $\bar{\pi}^{p,0}_{\Lambda_L}$.
Split $\Lambda_L$ into a boundary layer and an interior:
\[
\Lambda_{\mathrm{bd}}
:=\{v\in\Lambda_L:\ \mathrm{dist}(v,\partial\Lambda_L)\le \log|\Lambda_L|\},
\qquad
\Lambda_{\mathrm{int}}:=\Lambda_L\setminus\Lambda_{\mathrm{bd}}.
\]
By Proposition~\ref{LDP_domain} (boundary estimate with $h=n_+/2$) there exists $c=c(\beta,p)$ such that, for all $v\in\Lambda_{\mathrm{bd}}$,
\[
\bar{\pi}^{p,0}_{\Lambda_L}(\eta_v>n_+) \;\le\; c\,e^{-2\beta n_+}.
\]
As $|\partial\Lambda_L|=O(L)$ and $\log|\Lambda_L|=\log(L^2)=2\log L$, we obtain
\begin{equation}\label{eq:bd-sum}
\sum_{v\in\Lambda_{\mathrm{bd}}}\bar{\pi}^{p,0}_{\Lambda_L}(\eta_v>n_+)
\;\le\; c\,L\log L\,e^{-2\beta n_+}.
\end{equation}
For $v\in\Lambda_{\mathrm{int}}$, Proposition~\ref{LDP_domain} (bulk estimate) and monotonicity give
\[
\bar{\pi}^{p,0}_{\Lambda_L}(\eta_v>n_+)
\;\le\;
\hat{\pi}^{p}(\eta_0\ge n_+/2)\;+\;|\Lambda_L|^{-2\beta}.
\]
Summing over $|\Lambda_{\mathrm{int}}|\le L^2$ sites yields
\begin{equation}\label{eq:int-sum}
\sum_{v\in\Lambda_{\mathrm{int}}}\bar{\pi}^{p,0}_{\Lambda_L}(\eta_v>n_+)
\;\le\; L^2\,\hat{\pi}^{p}(\eta_0\ge n_+/2)\;+\;L^2\cdot |\Lambda_L|^{-2\beta}
\;=\; L^2\,\hat{\pi}^{p}(\eta_0\ge n_+/2)\;+\;L^{-4\beta+2}.
\end{equation}
Combining \eqref{eq:bd-sum} and \eqref{eq:int-sum},
\[
\sum_{v\in\Lambda_L}\bar{\pi}^{p,0}_{\Lambda_L}(\eta_v>n_+)
\;\le\;
c\,L\log L\,e^{-2\beta n_+}
\;+\;
L^2\,\hat{\pi}^{p}(\eta_0\ge n_+/2)
\;+\;
L^{-4\beta+2}.
\]
Plugging this into \eqref{eq:union} and using $(1-x)^{-1}\le 1+2x$ for $x$ small (here ensured by sufficiently large $\beta$ and $n_+\geq \log L$) gives
\[
\pi^{p,0}_{\Lambda_L}(A)
\;\le\;
\Bigl[
1+2\bigl(
c\,L\log L\,e^{-2\beta n_+}
+L^2\,\hat{\pi}^{p}(\eta_0\ge n_+/2)
+L^{-4\beta+2}
\bigr)
\Bigr]\,
\bar{\pi}^{p,0}_{\Lambda_L}(A).
\]
Together with $\bar{\pi}^{p,0}_{\Lambda_L}(A)\le \pi^{p,0}_{\Lambda_L}(A)$, this proves the proposition.
\end{proof}

% Proof of main theorem
\section{Proofs of Theorem~\ref{maintheorem} and Theorem~\ref{maintheorem2}}\label{sec:maintheorem}
\subsection{Proof of Theorem ~\ref{maintheorem}}
\begin{proof}[Proof of Theorem~\ref{maintheorem}]\label{mainthmpf}

We aim to prove that
\[
\exp\left({cL^{1-o(1)}}\right) \;\le\; \tmix.
\]
In order to establish this lower bound, it suffices to show that the dynamics requires exponentially long time to reach the typical region of configurations.  
By Theorem~\ref{thm-h-event}, let
\[
h = H(p,L) - 2,
\]
where $K$ is the constant appearing in Theorem~\ref{thm-h-event}. Define the event
\begin{align*}
    \mathcal{B} \;&:=\; \left\{ \eta : \#\{ x \in \Lambda_L : \eta_x \ge h+1 \} \,\ge\, \tfrac{1}{2} L^2 \right\}\\
    &=\; \left\{ \eta : \#\{ x \in \Lambda_L : \eta_x \ge H(p, L)-1 \} \,\ge\, \tfrac{1}{2} L^2 \right\}
\end{align*}
and let $\tau_{\mathcal{B}}$ denote the hitting time of $\mathcal{B}$. Since
\begin{equation}
\label{eq-31}
\pi_\Lambda^{p,0}(\mathcal{B}) = 1 - o(1),
\end{equation}
we conclude that if
\[
\tau_{\mathcal{B}} \;\ge\; \exp{\left(cL^{1-o(1)}\right)},
\]
then necessarily $\tau_{\mathrm{mix}} \ge \exp{\left(cL^{1-o(1)}\right)}$. Thus, to complete the proof of the mixing-time lower bound, it suffices to show that there exists a sufficiently small constant $c>0$ such that
\begin{equation}
\label{eq-3100}
\min_{\eta} \, \mathbb{P}^\eta\!\left( \tau_{\mathcal{B}} < \exp{\left(cL^{1-o(1)}\right)} \right) = o(1).
\end{equation}
In other words, starting from the worst possible initial configuration, the dynamics will, with high probability, require at least exponential time to reach the typical region $\mathcal{B}$.\\
Since $\mathcal{B}$ is an increasing event, we can reduce the problem to analyzing the dynamics starting from the minimal configuration $\sqcup$ (all heights equal to zero). Indeed, for any initial law $\nu$ we have
\begin{align*}
    \min_\eta \, \mathbb{P}^\eta\!\left( \tau_{\mathcal{B}} < \exp{\left(cL^{1-o(1)}\right)} \right)&= \mathbb{P}^\sqcup\!\left( \tau_{\mathcal{B}} < \exp{\left(cL^{1-o(1)}\right)} \right)\\
&\;\le\; \mathbb{P}^\nu\!\left( \tau_{\mathcal{B}} < \exp{\left(cL^{1-o(1)}\right)} \right),
\end{align*}
where the equality follows from the monotonicity of the dynamics with respect to the initial condition.
To prove the lower bound on the mixing time, it suffices to show that reaching the typical set
$\mathcal{B}$ requires the creation of a large upward fluctuation, which is exponentially unlikely.
To capture this mechanism, we introduce an atypical set $A$ that excludes configurations
containing macroscopic $h$-contours.

Specifically, for a parameter $\delta\in(0,\tfrac14)$ to be fixed later, define
\begin{equation}
\label{atypical_set}
A
\;:=\;
\bigcap_{\gamma:\,|\Lambda_\gamma|>\delta L^2}
\bigl(\mathcal{C}_{\gamma,h}\bigr)^c,
\end{equation}
that is, $A$ consists of configurations for which no $h$-contour encloses a region of area larger
than $\delta L^2$.
In other words, $A$ describes configurations where the interface remains close to the floor and
does not contain droplets occupying a positive fraction of the system.

The role of the set $A$ is the following.
Any trajectory that reaches the typical region $\mathcal{B}$ must necessarily create a macroscopic
$h$-contour at some intermediate time.
Such a contour carries an energetic cost proportional to its perimeter, and the probability of
its spontaneous formation is therefore exponentially small.
Consequently, starting from configurations in $A$, the dynamics is forced to overcome a large
energetic barrier before it can enter $\mathcal{B}$, leading to an exponential delay.
\begin{comment}
    Having defined the atypical set $A$, we now choose the initial distribution $\nu$ as the equilibrium
measure conditioned on $A$.
To make this precise, let $\delta=\delta_L$ be a positive sequence tending to zero:
\begin{equation}
\label{eq-25-pre}
\delta <
\left[
\left(\beta-\log\mu_2-o(L)\right)
\cdot
\frac{2\bar{\pi}_\Lambda^{p,0}(\eta_0\ge H(p, L)+1)}
{5\beta\,\bar{\pi}_\Lambda^{p,0}(\eta_0\ge h)}
\right]^2,
\end{equation}
where
\[
o(L):=
c_1\log(2(L+1)^2)\cdot
\exp\!\left(-4\beta c_1\sqrt{\frac{\log L}{\beta}}\right),
\]
and $\mu_2$ denotes the connective constant of $\mathbb{Z}^2$. As shown in Appendix~\ref{app:delta}, for all $p>1$ the ratio of one-point tails appearing above is of order $L^{-o(1)}$, and consequently the parameter $\delta_L$ can be chosen so that $\delta_L=L^{-o(1)}$.
\end{comment}

Having defined the atypical set $A$, we now choose the initial distribution $\nu$ as the equilibrium measure conditioned on $A$. To make this precise, let $\delta = \delta_L$ be a positive sequence tending to zero:
\begin{equation}
\label{eq-25-pre}
\delta <
\left[
\left(\beta-\log\mu_2-o(L)\right)
\cdot
\frac{2\bar{\pi}_\Lambda^{p,0}(\eta_0\ge H(p, L)+1)}
{5\beta\,\bar{\pi}_\Lambda^{p,0}(\eta_0\ge h)}
\right]^2,
\end{equation} where $o(L) := c_1 \log(2(L + 1)^2) \cdot \exp\left( -4\beta c_1 \sqrt{\frac{\log L}{\beta}} \right)$. The ratio of one-point tails in \eqref{eq-25-pre} can be estimated using the large deviation results established in the literature. Prior to the estimate, by standard comparison estimates between finite- and infinite-volume measures
(see Proposition~\ref{LDP_domain}), this ratio has the same asymptotic
behavior as
\[
\frac{\hat{\pi}^{p}(\eta_0\ge H(p,L)+1)}
{\hat{\pi}^{p}(\eta_0\ge H(p,L)-2)}.
\]
When $p=2$, we apply the precise height probability characterization from Theorem 3.1 of \cite{lubetzky2016harmonic} (Eq. 3.1):$$\begin{aligned} \frac{\hat{\pi}^2(\eta_0 = H(2,L)+1)}{\hat{\pi}^2(\eta_0 = H(2,L)-2)} &= \prod_{j=-1}^{2} \frac{\hat{\pi}^2(\eta_0 = H(2,L)-j+1)}{\hat{\pi}^2(\eta_0 = H(2,L)-j)}\\ &\le \exp \left( - \sum_{j=-1}^{2} \frac{c \beta (H-j+1)}{\log(H-j+1)} \right) \\ &\approx \exp \left( - \frac{3 c \beta H}{\log H} \right) = L^{-o(1)}, \end{aligned}$$where the last equality follows from the asymptotics $H(2, L) \sim \sqrt{\frac{\log L \log \log L}{2\pi\beta}}$. For the generalized $p$-SOS models ($p \neq 2$), we instead invoke Theorem 7.1 of \cite{chen2025limiting} (Eq. 7.4), which provides a similar uniform bound for all $p \in (1, \infty) \setminus \{2\}$:
$$\begin{aligned} \frac{\hat{\pi}^{p}(\eta_0 = H(p,L)+1)}{\hat{\pi}^{p}(\eta_0 = H(p, L)-2)} 
&\le \exp \left( - c \beta \sum_{j=-1}^{2} (H(p, L)-j)^{(p-1) \wedge 1} \right) \\
&\approx \exp \left( - 3 c \beta H(p, L)^{(p-1) \wedge 1} \right) \\ 
&= \begin{cases} \exp\left(-O\left(\left(\log L)^{\frac{p-1}{p}}\right)\right)\right) & \text{for } 1 < p < 2 \\ \exp\left(-O\left(\sqrt{\log L}\right)\right) & \text{for } 2 < p < \infty \end{cases} = L^{-o(1)} \end{aligned}$$This confirms that the tail ratio is consistently of order $L^{-o(1)}$ for all $p > 1$. Consequently, the parameter $\delta_L$ can be chosen as $\delta_L = L^{-o(1)}$.

In other words, we restrict the initial configuration to those without macroscopic $h$-contours, making it entropically unlikely to reach the target set $\mathcal{B}$ in subexponential time.\\
In the sequel, we denote by $\partial A$ the internal boundary of $A$, defined as
\[
\partial A := \left\{ \eta \in A : p(\eta,\eta') > 0 \text{ for some } \eta' \notin A \right\},
\]
where $p(\cdot,\cdot)$ is the transition probability of the dynamics.\\
Intuitively, $\partial A$ represents the set of configurations in $A$ that can exit $A$ in a single update.
Let $\tau_{\partial A}$ be the hitting time of $\partial A$.

Up to the hitting time $\tau_{\partial A}$, if the chain starts from a configuration in $A \setminus \partial A$, its evolution is identical to that of a Glauber dynamics constrained to remain in $A$, whose reversible distribution is exactly $\nu = \pi_\Lambda^{p,0}(\cdot \mid A)$.  
Applying a union bound over integer times $0 \le t \le \exp{\left(cL^{1-o(1)}\right)}$, we arrive at
\begin{align}
\label{eq:union-bound}
&\mathbb{P}^\nu\!\left( \tau_\mathcal{B} < \exp{\left(cL^{1-o(1)}\right)} \right)\\&\le \mathbb{P}^\nu\!\left( \tau_{\partial A} < \exp{\left(cL^{1-o(1)}\right)} \right)
   + \mathbb{P}^\nu\!\left( \tau_\mathcal{B} < e^{cL},\ \tau_{\partial A} \ge \exp{\left(cL^{1-o(1)}\right)} \right) \\
&\le \exp{\left(cL^{1-o(1)}\right)} \bigl[ \nu(\partial A) + \nu(\mathcal{B}) \bigr]. \nonumber
\end{align}
We next establish exponential decay for the two terms inside the bracket.
To keep the exposition concise, we defer the technical estimates to the lemmas below.
Throughout, we retain the notation introduced earlier (in particular, $A,\tilde A, \delta, \varepsilon_L$ defined above). 
The following two lemmas summarize the required bounds on the boundary term $\nu(\partial A)$ 
and on the probability of the typical set $\nu(\mathcal{B})$.

\begin{lemma}
\label{lem:boundary-term}
For sufficiently large $\beta$, there exists $c_2 = c_2(\beta) > 0$ such that
\begin{align*}
    \nu(\partial A) \;\le\;  \exp{\left(-c_2L^{1-o(1)}\right)}
\end{align*}
\end{lemma}

\begin{lemma}\label{probB}
For sufficiently large $\beta$, there exists $c_3 = c_3(\beta) > 0$ such that
\[
\nu(\mathcal{B}) \;\le\; e^{-c_3 L}.
\]
\end{lemma}

From Lemma~\ref{lem:boundary-term}, we already know that the boundary contribution $\nu(\partial A)$ decays exponentially fast in $L$.  
Similarly, Lemma~\ref{probB} ensures the same exponential decay for $\nu(\mathcal{B})$.  
Hence, both terms inside the bracket in \eqref{eq:union-bound} are exponentially small, and it remains to translate these bounds from the auxiliary floor-only measure to the original measure with both floor and ceiling.

Recall that for any measurable 
$A \subset \{0,\dots,n_+\}^{\Lambda}$,
Proposition~\ref{comparisonwithnoandyes} yields
\begin{equation}\label{eq:comparison-floor-ceiling}
\bar{\pi}^{p,0}_\Lambda(A)
\le
\pi^{p,0}_\Lambda(A)
\le
\Big[
1+2\big(
c_1|\gamma|\log|\gamma| e^{-2\beta n_+}
+L^2\,\hat{\pi}^{p}(\eta_0 \ge n_+/2)
+|\Lambda|^{-2\beta+1}
\big)
\Big]
\bar{\pi}^{p,0}_\Lambda(A),
\end{equation}
where $\bar{\pi}^{p,0}_\Lambda$ denotes the floor-only $p$-SOS measure. In the regime $\log L \le n_+$ and for sufficiently large $\beta$,
the correction terms in~\eqref{eq:comparison-floor-ceiling}
are negligible. Indeed,
\begin{equation}\label{eq:ceiling-error}
c_1|\gamma|\log|\gamma| e^{-2\beta n_+}
\le
2c_1(L+1)^2 \log\!\big[2(L+1)^2\big]
\, e^{-2\beta \log L}
\le
L^3 L^{-2\beta}
\le
L^{-2}.
\end{equation}

To bound the remaining term, we require an estimate on the
one-point tail probability $\hat{\pi}^{p}(\eta_0 \ge h)$.
Since the large-deviation behavior differs in the regimes
$1<p<2$, $p=2$, and $2<p<\infty$, the proof is carried out
case by case and deferred to Appendix~\ref{Appendix1}.
There we show that
\begin{equation}\label{eq:onepoint-tail}
\hat{\pi}^{p}(\eta_0 \ge n_+/2) \le L^{-3}.
\end{equation}

Thus, for large enough $\beta$,
\[
\pi^{p,0}_\Lambda(A) = \big(1+o(1)\big)\,\bar{\pi}^{p,0}_\Lambda(A).
\]
As a consequence, any exponentially small probability bound established under $\bar{\pi}^{p,0}_\Lambda$In particular, this implies that all exponentially small probability bounds established under $\bar{\pi}^{p,0}_\Lambda$ carry over directly to the full measure $\pi^{p,0}_\Lambda$.
In conclusion, the second term on the right-hand side of~\eqref{eq:union-bound} is $o(1)$ whenever $c<\min(c_2,c_3)$, completing the proof.
\end{proof}
We begin by estimating the first term in equation~\eqref{eq:union-bound}.
The following lemma provides the required bound.

\begin{proof}[Proof of Lemma~\ref{lem:boundary-term}]
Next, introduce the set
\begin{equation}\label{badevent}
    \tilde{A} := \bigcap_{\gamma:\ |\Lambda_\gamma| > \frac{\delta}{5} L^2} \left(\mathcal{C}_{\gamma,h} \right)^c,
\end{equation}
which consists of configurations with no $h$-contour of area larger than $\delta L^2/5$.  
Since modifying a single site can create at most four disjoint $h$-contours, it follows that $\partial A \subset A \setminus \tilde{A}$.  
Consequently,
\[
\nu(\partial A) = \frac{\pi_\Lambda^{p,0}(\partial A)}{\pi_\Lambda^{p,0}(A)}
\ \le\ \frac{\pi_\Lambda^{p,0}(A \setminus \tilde{A})}{\pi_\Lambda^{p,0}(A)}.
\]

%% 여기까지
We now aim to bound the ratio
\begin{equation}\label{ratioeq}
    \frac{ \pi_\Lambda^{p,0}(A \setminus \tilde{A}) }{ \pi_\Lambda^{p,0}(A) }
\end{equation}
and show that it decays exponentially in $L$.  
To this end, consider a contour $\gamma$ with relative area $|\Lambda_\gamma|/L^2 \in \left(\frac{\delta}{5}, \delta\right)$.  
Define the map $T_\gamma$ as in~(\ref{mapping}), acting on a configuration $\eta$ by lowering every height inside $\gamma$ by one unit, i.e.,
\[
(T_\gamma \eta)_v =
\begin{cases}
\eta_v - 1, & v \in \Lambda_\gamma, \\
\eta_v, & \text{otherwise}.
\end{cases}
\]
Whenever $\eta \in A \cap \mathcal{C}_{\gamma,h}$, the transformed configuration $T_\gamma \eta$ still belongs to $A$, which implies
\[
\pi_\Lambda^{p,0}(A) \ge
\sum_{\eta \in A \cap C_{\gamma,h},\ \eta|_{\Lambda_\gamma} > 0}
\pi_\Lambda^{p,0}(T_\gamma \eta).
\]
In the proof of Theorem~\ref{prop4.1}, by equation~(\ref{eq:peierls-num}), this sum can be bounded above by
\begin{equation}\label{per-1}
    e^{\beta|\gamma|} \cdot \pi_\Lambda^{p,0}(\eta|_{\Lambda_\gamma} > 0 \mid A, \mathcal{C}_{\gamma,h}) \cdot
\pi_\Lambda^{p,0}(A \cap \mathcal{C}_{\gamma,h}).    
\end{equation}
It remains to control the middle factor.  
Similarly, in the proof of Theorem~\ref{prop4.1}, from equation~(\ref{eq:pos-prob-factor1}) we obtain a product-form lower bound, obtained via independence across sites and uniform one-point estimates, which yields
\begin{equation}\label{per-2}
    \pi_\Lambda^{p,0}(\eta|_{\Lambda_\gamma} > 0 \mid \mathcal{C}_{\gamma,h})
\ \ge\ \exp\left( -\, |\Lambda_\gamma|\, \bar{\pi}^p(\eta_0 \ge h)
- c_1\, |\gamma| \log|\gamma|\, e^{-4\beta h} \right).
\end{equation}

Combining (\ref{per-1}), (\ref{per-2}), we deduce
\begin{equation}\label{eq-29}
\pi_\Lambda^{p,0}(\mathcal{C}_{\gamma,h} \mid A)
\le \exp\left( -\beta |\gamma| + |\Lambda_\gamma| \bar{\pi}^p(\eta_0 \ge h)
+ c_1 |\gamma| \log|\gamma|\, e^{-4\beta h} \right),
\end{equation}
which is the desired contour probability bound.

% 여기까지
By the planar isoperimetric inequality on $\mathbb{Z}^2$, any contour $\gamma$ satisfies
\[
|\Lambda_\gamma| \le \frac{|\gamma|^2}{16}.
\]
When $|\Lambda_\gamma| < \delta L^2$, this yields
\begin{equation}\label{first bound}
\begin{aligned}
|\Lambda_\gamma| \cdot \bar{\pi}^p(\eta_0 \ge h) 
&\le \cdot \sqrt{ \delta L^2 } \cdot \sqrt{ \frac{|\gamma|^2}{16} } \cdot \bar{\pi}^p(\eta_0 \ge h) \\
&= \sqrt{\delta} \cdot \frac{1}{2} L |\gamma| \cdot \bar{\pi}^p(\eta_0 \ge h) \\
&\le \left[ \left( \beta - \log \mu_2 - \varepsilon_L \right) \cdot \frac{2 \bar{\pi}^p(\eta_0 \ge H+1)}{5\beta  \bar{\pi}^p(\eta_0 \ge h)} \right] \cdot L |\gamma| \cdot \frac{1}{2} \cdot \bar{\pi}^p(\eta_0 \ge h) \\
&\le \left( \beta - \log \mu_2 - \varepsilon_L \right) \cdot \frac{5\beta}{5\beta L} \cdot L |\gamma| \\
&= \left( \beta - \log \mu_2 - \varepsilon_L \right) |\gamma|.
\end{aligned}
\end{equation}
On the other hand, recall that $\log L\leq n_+$, so that
\begin{equation}\label{second bound}
\begin{aligned}
c_1 |\gamma| \log |\gamma| \cdot e^{-4\beta h}
&\le c_1 |\gamma| \log \left( 2(L+1)^2 \right) \cdot e^{-4\beta h} \quad (\text{since } |\gamma| \le 2(L+1)^2) \\
&\le c_1 |\gamma| \log \left( 2(L+1)^2 \right) \cdot e^{-4\beta \log L}\\
&= c_1 |\gamma| \cdot \exp\left( \log \log \left( 2(L+1)^2 \right) - 4\beta \log L )\right) \\
&= o(|\gamma|).
\end{aligned}
\end{equation}
Hence, the right-hand side of~\eqref{eq-29} is bounded above by $e^{-(\beta - \lambda)|\gamma|}$.\\
To upper bound of ~\eqref{ratioeq}, we can modify that ratio as below. Due to the definition of ~\eqref{badevent}, we get
\begin{align*}
    A^{\mathrm c}\bigcup_{\gamma:\,|\Lambda_\gamma|>\frac15\delta L^2} C_{\gamma,h},\quad 
A\setminus \tilde A
&= A\cap \tilde A^{\mathrm c}
\subset \bigcup_{\gamma:\,|\Lambda_\gamma|>\frac15\delta L^2} (A\cap C_{\gamma,h}),
\end{align*}
Therefore we get the upper bound of ~\eqref{ratioeq} as
\begin{equation}\label{ratioupper}
    \begin{aligned}
    \frac{\pi_\Lambda^0(A\setminus\tilde A)}{\pi_\Lambda^0(A)}
    &\leq
    \sum_{\gamma:\,|\Lambda_\gamma|>\frac15\delta L^2}\frac{\pi_\Lambda^0(A\cap C_{\gamma,h})}{\pi_\Lambda^0(A)}
    \\&=
    \sum_{\gamma:\,|\Lambda_\gamma|>\frac15\delta L^2}\pi_\Lambda^0(C_{\gamma,h}\mid A).
    \end{aligned}
\end{equation}
Then we already get the upper bound of ~\eqref{eq-29}, so we get

\begin{align*}
\frac{\pi_\Lambda^0(A\setminus\tilde A)}{\pi_\Lambda^0(A)}
&\le
\sum_{\gamma:\,|\Lambda_\gamma|>\frac15\delta L^2}
e^{-(\beta-\lambda)|\gamma|}.
\end{align*}
We now bound the right-hand side by summing over all contours with sufficiently large area.
Using the isoperimetric inequality, any contour $\gamma$ with $|\Lambda_\gamma|>(\delta/5)L^2$
must have length $|\gamma|\ge c\sqrt{\delta}\,L$ where $c$ is a constant that satisfying $0<c<\sqrt{16/5}$.
Moreover, the number of contours of length $k$ is at most $L^2\mu_2^k$.
Therefore,
\begin{equation}
\begin{aligned}\label{isoandpeier}
    \sum_{\gamma:\,|\Lambda_\gamma|>\frac15\delta L^2}
e^{-(\beta-\lambda)|\gamma|}
&\leq
\sum_{k\ge c\sqrt{\delta}L}
\#\{\gamma:\,|\gamma|=k\}\,e^{-(\beta-\lambda)k}
\\
&\leq
\sum_{k\ge c\sqrt{\delta}L}
L^2\,\mu_2^{\,k}\,e^{-(\beta-\lambda)k}
\\
&=
L^2\sum_{k\ge c\sqrt{\delta}L}
e^{-(\beta-\lambda-\log\mu_2)k}
\\
&\leq
\exp(-c_2\sqrt{\delta}\,L),
\end{aligned}
\end{equation}

In conclusion, the above estimate shows that the contribution of contours
with $|\Lambda_\gamma|>\tfrac15\delta L^2$ is exponentially small, and in particular the first term in the right-hand side of~\eqref{eq:union-bound} is $o(1)$ provided $c<c_1$. This verifies that the union bound in~\eqref{eq:union-bound} is valid, and therefore completes the proof of Lemma~\ref{lem:boundary-term}.
\end{proof}
We next turn to the second term in equation~\eqref{eq:union-bound}.
The following lemma establishes the corresponding bound.
\begin{proof}[Proof of Lemma~\ref{probB}] First, we claim that for any short $h$-contour $\gamma$ and any $v \in \Lambda_\gamma$, where short means that $|\gamma| < \log^2 L$, we have
\begin{equation}
\label{eq-h+1-given-C(gamma,h)}
\bar{\pi}^{p,0}_\Lambda \left( \eta_v \ge h+1 \mid \mathcal{C}_{\gamma,h} \right) \le \tfrac{1}{4}.
\end{equation}
Indeed, let $\mathscr{A} = \{ \gamma' : v \in \Lambda_{\gamma'} \subseteq \Lambda_\gamma \}$.
Then, by applying~\eqref{mapping} from Lemma~\ref{peierlsimple}, we obtain
\[
\hat{\pi}^{p,0}_\Lambda \left( \eta_v \ge h+1 \mid \mathcal{C}_{\gamma,h} \right)
\le \sum_{\gamma' \in \mathscr{A}} \hat{\pi}^{p,0}_\Lambda \left( C_{\gamma', h+1} \mid \mathcal{C}_{\gamma,h} \right)
\le \sum_{\gamma' \in \mathscr{A}} e^{-\beta |\gamma'|} \le \tfrac{1}{8},
\]
for $\beta$ sufficiently large, since the number of contours $\gamma' \in \mathscr{A}$ of length $k$ is at most $k \mu_2^k$, using the fact that each such $\gamma'$ crosses a horizontal line within distance $k$ to the right of $v$.
To transfer this estimate to the setting with a floor, observe that by Remark 3.8 in~\cite{caputo2014dynamics}, we have
\begin{equation}
\label{eq-h+1-given-C(gamma,h)-intermediate}
\bar{\pi}^{p,0}_\Lambda \left( \eta_v \ge h+1 \mid \mathcal{C}_{\gamma,h} \right)
= \frac{
\hat{\pi}^{p,0}_\Lambda \left( \eta_v \ge h+1,\ \eta|_{\Lambda_\gamma} \ge 0 \mid \mathcal{C}_{\gamma,h} \right)
}{
\hat{\pi}^{p,h}_{\Lambda_\gamma} \left( \eta|_{\Lambda_\gamma} \ge 0 \mid \eta|_{\partial_\gamma} \ge h \right)
}.
\end{equation}
We have already established that the numerator is at most $1/8$. For the denominator, by monotonicity,
\[
\hat{\pi}^{p, h}_{\Lambda_\gamma} (\eta\restriction_{\Lambda
_{\gamma}}
\geq0 ) = \hat{\pi}^{p, 0}_{\Lambda_\gamma} (\eta\restriction_{\Lambda
_{\gamma}}
\geq-h ) \geq1 - \hat{\pi}^{p, 0}(\eta\geq0)|\Lambda_\gamma|-|\gamma|\log|\gamma|e^{-4\beta h}.
\]
using Proposition~\ref{LDP_domain} and a union bound over the sites in $\Lambda_\gamma$.
Since $|\Lambda_\gamma| \le |\gamma|^2 = O(\log^4 L)$, the last term is $L^{-1+o(1)}$, so the denominator in~\eqref{eq-h+1-given-C(gamma,h)-intermediate} is at least $7/8$.

Hence,
\[
\bar{\pi}^{p,0}_\Lambda \left( \eta_v \ge h+1 \mid \mathcal{C}_{\gamma,h} \right) \le \frac{1/8}{7/8} \le \tfrac{1}{4},
\]
as claimed in~\eqref{eq-h+1-given-C(gamma,h)}.
%% 여기까지
With inequality~\eqref{eq-h+1-given-C(gamma,h)} at hand, we now establish a concentration bound
for the contribution of short $h$-contours.
Specifically, we show that the total number of sites at height at least $h+1$ contained in
the interiors of short contours is sharply concentrated.

More precisely, we prove that
\begin{equation}\label{concentration_inequality1}
\bar{\pi}^{p,0}_\Lambda \left(
\sum\nolimits'_\gamma \#\left\{ x \in \Lambda_\gamma : \eta_x \ge h+1 \right\}
\ge \tfrac{1}{2} L^2
\right)
\le \exp\left( -c_2 L^{2 - o(1)} \right),
\end{equation}
where the summation $\sum\nolimits'_\gamma$ ranges over all short $h$-contours $\gamma$.
For each admissible $h$-contour $\gamma$ (indexed in the primed sum), define the random variable
\[
X_\gamma := \#\left\{ x \in \Lambda_\gamma : \eta_x \ge h+1 \right\}.
\]
Conditioned on the event that all relevant $\gamma$ are $h$-contours (i.e., $\eta_x \ge h$ on $\Delta^+_\gamma$), the random variables $\{X_\gamma\}_\gamma$ are mutually independent, since the associated domains $\Lambda_\gamma$ are disjoint.
By inequality~\eqref{eq-h+1-given-C(gamma,h)}, the expectation of each $X_\gamma$ satisfies
\[
\mathbb{E}[X_\gamma] \le \tfrac{1}{4} |\Lambda_\gamma|.
\]
Furthermore, since $X_\gamma$ counts the number of sites in $\Lambda_\gamma$ where $\eta_x \ge h+1$ and each height $\eta_x$ is trivially bounded by $\log L$ due to the ceiling, we have the uniform bound
\[
X_\gamma \le |\Lambda_\gamma| \le (\log L)^4,
\]
where we used that $|\gamma| \le \log^2 L$ implies $|\Lambda_\gamma| \le |\gamma|^2 \le (\log L)^4$.
Letting $Y := \sum\nolimits'_\gamma X_\gamma$ be the total number of sites (across all such contours) at height at least $h+1$, we observe
\[
\mathbb{E}[Y] \le \tfrac{1}{4} \sum |\Lambda_\gamma| \le \tfrac{1}{4} L^2.
\]
Since the $\{X_\gamma\}$ are independent and uniformly bounded by $(\log L)^4$, we may apply Hoeffding's inequality to obtain
\[
\mathbb{P}(Y \ge \tfrac{1}{2} L^2)
\le \exp\left( - \frac{2 \left( \tfrac{1}{4} L^2 \right)^2 }{ N (\log L)^8 } \right),
\]
where $N \le L^2$ is the number of disjoint $h$-contours in the sum.
%% 여기까지
Thus,
\[
\mathbb{P}\left( Y \ge \tfrac{1}{2} L^2 \right)
\le \exp\left( - \frac{L^4}{8 L^2 (\log L)^8} \right)
= \exp\left( - \frac{L^2}{8 (\log L)^8} \right).
\]
Replacing constants appropriately, we recover inequality~(\ref{concentration_inequality1}), where the summation $\sum'$ is over all short $h$-contours $\gamma$.\\
Next goal in this part is to control the contribution of long $h$-contours.
More precisely, we will prove that there exists a constant $c_3>0$ such that
\begin{equation}\label{concentration_inequality2}
    \bar{\pi}_\Lambda^{p,0} \left(
    \sum_{\gamma \in B}'' \#\left\{ v \in \Lambda_\gamma : \eta_v \ge h+1 \right\}
    \ge \tfrac{1}{2} L^2
    \right)
    \le \exp(-c_3 L),
\end{equation}
where the summation $\sum''$ is over all long $h$-contours $\gamma$, namely those with
$|\gamma|>\log^2 L$.
Let $B = B(\eta)$ denote the collection of $h$-contours $\gamma$ of a configuration $\eta$ such that $|\gamma| > \log^2 L$:
\begin{align*}
    B := \left\{ \gamma : \gamma \text{ is an } h\text{-contour of } \eta \text{ with } |\gamma| > \log^2 L \right\}.
\end{align*}
Suppose that for some $\theta > \log \mu_2$, we have the following uniform bound for all such contours $\gamma$ with small area, as ensured by an inequality of the form~\eqref{eq-29}:
\[
\bar{\pi}^{p,0}_\Lambda(\mathcal{C}_{\gamma,h} ) \le \exp( -\theta |\gamma| ).
\]
Indeed, in our setting, combining the estimate
\[
\bar{\pi}^{p,0}_\Lambda \left(\mathcal{C}_{\gamma,h} \mid A \right)
\le \exp\left( -\beta |\gamma| + \bar{\pi}^p(\eta_0 \ge h) |\Lambda_\gamma| + 2c |\gamma| \log |\gamma| \cdot e^{-4\beta h} \right)
\]
with the isoperimetric inequality $|\Lambda_\gamma| \le \tfrac{|\gamma|^2}{16}$ and applying bounds~\eqref{first bound} and~\eqref{second bound}, we obtain that for contours $\gamma$ with $|\Lambda_\gamma| \le \delta L^2$ and for $\delta > 0$ sufficiently small,
\[
\bar{\pi}^{p,0}_\Lambda \left(\mathcal{C}_{\gamma,h} \mid A \right)
\le \exp\left( -(\beta - \lambda) |\gamma| \right)
= \exp( -\theta |\gamma| ),
\]
where $\theta := \beta - \lambda > \log \mu_2$ provided that $\beta$ is sufficiently large.
%% 여기까지
Let $\gamma_1, \dots, \gamma_m$ be $m$ disjoint $h$-contours with $|\gamma_i| > \log^2 L$ and disjoint interiors $\{ \Lambda_{\gamma_i} \}_{i=1}^m$.
Then, by independence over disjoint supports and the exponential bound on each contour, we have:
\begin{equation}\label{eq:33}
\bar{\pi}_\Lambda^{p,0} \left( \bigcap_{i=1}^m C_{\gamma_i, h} \right)
\le \exp\left( -\theta \sum_{i=1}^m |\gamma_i| \right).
\end{equation}
We now bound the total probability that such a collection of $m$ disjoint long $h$-contours appears, by summing over all such contour configurations.
For each $\gamma_i$ of length $k_i = |\gamma_i|$, there are at most $L^2$ possible choices for its origin (starting point) and at most $R^{k_i}$ possible choices for a self-avoiding path of length $k_i$.
Therefore, for a fixed $m$-tuple of contour lengths $(k_1, \dots, k_m)$ with each $k_i > \log^2 L$, the number of such contour configurations is at most
\[
\prod_{i=1}^m L^2 \mu_2^{k_i}.
\]
Combining this with the exponential probability bound~\eqref{eq:33}, we obtain
\[
\sum_{k_1, \dots, k_m > \log^2 L}^{L^2} \left( \prod_{i=1}^m L^2 \mu_2^{k_i} e^{-\theta k_i} \right)
= \left( \sum_{k > \log^2 L}^{L^2} L^2 \mu_2^k e^{-\theta k} \right)^m.
\]
Finally, summing over all $m \ge 1$, we obtain the total probability bound:
\begin{align*}
\bar{\pi}_\Lambda^{p,0} \left( \exists \text{ disjoint long } h\text{-contours } \gamma_1, \dots, \gamma_m \right)
&\le \sum_{m=1}^\infty \left( \sum_{k > \log^2 L}^{L^2} L^2 \mu_2^k e^{-\theta k} \right)^m.
\end{align*}
Since $\theta > \log R$, the inner sum satisfies
\[
\sum_{k > \log^2 L}^{L^2} L^2 \mu_2^k e^{-\theta k} \le L^2 \sum_{k > \log^2 L} e^{-c k} = O\left( e^{-c \log^2 L} \right),
\]
for some $c = \theta - \log R > 0$.
Therefore, the entire series converges and satisfies
\begin{equation}\label{eq:4.10}
\sum_{m=1}^\infty \left( e^{-c \log^2 L} \right)^m = O\left( e^{-c \log^2 L} \right).
\end{equation}
We now choose the threshold $T$ to match the event in~\eqref{eq:4.10}, which corresponds to the number of vertices $v$ in the interiors of long $h$-contours $\gamma$ such that $\eta_v \ge h+1$ exceeding $\tfrac{1}{2} L^2$.
By the isoperimetric inequality, each such $\gamma$ satisfies $|\Lambda_\gamma| \le \frac{1}{16} |\gamma|^2$.
To ensure that the total area covered is at least $\tfrac{1}{2} L^2$, the total contour length must satisfy
\[
\sum_{\gamma \in B} |\gamma| \ge C L \quad \text{for some constant } C > 0.
\]
Therefore, we set $T = C L$ and deduce from the previous exponential tail bound:
\[
\bar{\pi}_\Lambda^{p,0} \left( \sum_{\gamma \in B} |\gamma| \ge T \right)
\le \exp(-c T) = \exp(-c C L).
\]

Combining this with the isoperimetric inequality again, we obtain that
\[
\sum_{\gamma \in B} \#\left\{ v \in \Lambda_\gamma : \eta_v \ge h \right\} \ge \tfrac{1}{2} L^2
\quad \Rightarrow \quad \sum_{\gamma \in B} |\gamma| \ge c' L,
\]
for some constant $c' > 0$. Hence,
\begin{equation}\label{concentration_inequality2}
    \bar{\pi}_\Lambda^{p,0} \left( \sum_{\gamma \in B}'' \#\left\{ v \in \Lambda_\gamma : \eta_v \ge h+1 \right\} \ge \tfrac{1}{2} L^2 \right)
\le \exp(-c_3 L),
\end{equation}
where the summation $\sum''$ is over all long $h$-contours $\gamma$, and $c_3 > 0$ is an absolute constant.
Hence, combining inequalities~\ref{concentration_inequality1} and~\ref{concentration_inequality2}, we obtain the bound~\ref{probB}.
\end{proof}

\subsection{Proof of Theorem~\ref{maintheorem2}}
Fix $a\in(0,1)$ and $p\in(1,\infty)$. Our goal is to prove an exponential lower bound on the hitting time
\[
\tau_{p,a}\ :=\ \min\{t\ge 0:\ \eta(t)\in\Omega_{p,a}\},
\qquad 
\Omega_{p,a}\ :=\ \bigl\{\eta:\ \#\{x:\eta_x\ge aH(p,L)\}\ge (1-\varepsilon_\beta)|\Lambda|\bigr\}.
\]
The argument parallels the proof of Theorem~\ref{maintheorem}, with the only difference that we now work at the intermediate height
\begin{align*}
    h:=aH(p,L)-1,
\end{align*}
so the initial law is supported on configurations forbidding \emph{large} $h$-contours at a lower level than in Theorem~\ref{maintheorem}. This lowers the critical droplet area scale, and the precise exponent depends on $p$.

Recall the definition
\begin{equation}\label{d(p)}
d(p):=\begin{cases}
p,& 1<p<2,\\
2,& 2\le p<\infty.
\end{cases}
\end{equation}
The exponent $d(p)$ arises from the energetic cost required to create a macroscopic droplet reaching height $aH(p,L)$. In the regime $1<p<2$, the Hamiltonian contains nonlinear gradient terms of the form $|\nabla \eta|^p$, and lifting the interface over a region of macroscopic area produces an energy cost governed by $p$-power scaling. In contrast, for $p\ge2$ the fluctuations are effectively Gaussian, and the dominant energetic contribution becomes quadratic. This change in fluctuation behavior leads to the transition at $p=2$ in the definition of $d(p)$ in~\eqref{d(p)}.
The exponent $d(p)$ plays a crucial role in the dynamical analysis: it determines the probability of creating the critical droplet required for the interface to reach an intermediate level $aH(p,L)$, and therefore directly controls the lower bound on the corresponding hitting time. Using the one-point large-deviation estimates summarized in Table~\ref{tab:HpL}, we obtain that for every $\varepsilon>0$ and all sufficiently large $L$ (for large inverse temperature $\beta$), \begin{equation}\label{eq:onepoint}
\hat\pi^p\big(\eta_0\ge aH(p,L)-1\big)
\le
L^{-a^{\,d(p)}+\varepsilon}.
\end{equation}
The derivation of~\eqref{eq:onepoint} requires a case-by-case analysis depending on the regimes $1<p<2$ and $2\le p<\infty$, since the large-deviation rate changes qualitatively at $p=2$. To avoid interrupting the main dynamical argument, the detailed proof is deferred to Lemma~\ref{one-point-lemma} and Appendix~\ref{Appendix2}.

\begin{lemma}\label{one-point-lemma}
Fix $a\in(0,1)$.
For sufficiently large $\beta$ and all $L$ large enough,
the infinite-volume $p$-SOS measure on $\mathbb Z^2$
satisfies the following one-point upper bounds:
\begin{enumerate}
\item If $1<p<2$, then
\[
\hat{\pi}^{p}\!\left(\eta_0 \ge aH(p,L)-1\right)
\le L^{-a^p+\varepsilon}.
\]
\item If $2\le p<\infty$, then
\[
\hat{\pi}^{p}\!\left(\eta_0 \ge aH(p,L)-1\right)
\le L^{-a^2+\varepsilon}.
\]
\end{enumerate}
\end{lemma}

\begin{proof}[Proof of Theorem ~\ref{maintheorem2}]
We aim to establish an exponential lower bound on the hitting time
\(\tau_{p,a}\), which depends on the parameters \(a\in(0,1)\) and
\(p\in(1,\infty)\). The proof closely follows that of
Theorem~\ref{maintheorem}. In particular, we will prove \eqref{eq-3100} under the modified choice
\[
h:=a\,H(p,L)-1.
\]
with this choice of $h$, Lemma~\ref{one-point-lemma} implies that for any $\varepsilon>0$ and all sufficiently large $L$,
\[
\hat\pi_p(\eta_0 \ge h)
=\hat\pi_p(\eta_0 \ge aH(p,L)-1)
\le L^{-a^{d(p)}+\varepsilon}.
\]
With this change, all ingredients of the proof remain the same as
in Theorem~\ref{maintheorem}. The only difference lies in the choice of the initial law for the Glauber dynamics.: We condition on the same atypical event \eqref{atypical_set} as before, but evaluated at the lower height level
\[
h:=aH(p,L)-1.
\]
Then, using the same parameter $\varepsilon$ as in Lemma~\ref{one-point-lemma},
we define $\delta$ by
\[
\delta=\delta_L:=L^{-2\varepsilon},
\]
and define the restricted event
\begin{equation}\label{eq:Apa}
A_{p,a}\;:=\;\bigcap_{\gamma:\ |\Lambda_\gamma|>\,\delta\,L^{\,2{a^{d(p)}}}}\big(\mathcal{C}_{\gamma,h}\big)^{\mathrm c},
\end{equation}
i.e., $A_{p,a}$ is the event that there exists no $h$-contour $\gamma$ whose enclosed area exceeds $\delta L^{2{a^{d(p)}}}$. This $\delta$ is tend to zero such that, for $|\Lambda|\leq\delta L^{2^{a(p)}}$;
\begin{align*}
    \,|\Lambda_\gamma|\,\hat\pi^p(\eta_0\ge h)
\;\le\;
\,\delta\,L^{2a^{d(p)}}\cdot L^{-a^{d(p)}+\varepsilon}
\;=\;
\,\delta\,L^{\,a^{d(p)}+\epsilon}.
\end{align*}
We start the dynamics from the equilibrium measure conditioned on $A_{p,a}$, namely $\nu:=\pi(\,\cdot\,|\,A_{p,a})$. The target set is
\[
\mathcal{B}_a\;:=\;\Bigl\{\eta:\ \#\{x:\ \eta_x\ge h+1\}\ge \tfrac12 L^2\Bigr\}.
\]
By the same argument as in \eqref{eq:union-bound}, we obtain
\begin{equation}\label{eq:std-decomp}
\mathbb{P}_\nu\!\big(\tau_{\mathcal{B}_a}<\exp{\left(c\,L^{\,{a^{d(p)}-o(1)}}\right)}
\;\le\;\exp{\left(c\,L^{\,{a^{d(p)}-o(1)}}\right)}\,\Bigl\{\nu(\partial A_{p,a})+\nu(\mathcal{B}_a)\Bigr\}.
\end{equation}
Consequently, if the right-hand side is $o(1)$, then $\tau_{\mathcal{B}_a}\ge \exp{\left(c\,L^{\,{a^{d(p)}-o(1)}}\right)}$ with probability $1-o(1)$. Since $\mathcal{B}_a\subset\Omega_{p,a}$ and $\pi(\Omega_{p,a})\rightarrow1$ at equilibrium, the desired lower bound follows.

First we begin by showing that
\[
\nu(\partial A_{p,a}) \le \exp\bigl(-c\,L^{\,a^{d(p)}-o(1)}\bigr).
\]
To this end, we introduce an auxiliary event $\widetilde A$, defined analogously to
\eqref{badevent}, and define
\[
A \,:=\, \bigcap_{\gamma:\,|\Lambda_\gamma|>\frac{\delta}{5}\,L^{a^{d(p)}}}
\bigl(C_{\gamma,h}\bigr)^{\mathrm c},
\]
where $h:=aH(p,L)-1$ and $\delta=\delta_L$ is chosen as in
\eqref{first bound}.

Our goal is to establish the upper bound in \eqref{ratioeq}.
By construction, the choice of $\delta=\delta_L$ ensures that
\eqref{first bound} holds, and \eqref{second bound} is also satisfied.
Therefore, we obtain
\[
\frac{\pi^{0}_{\Lambda}(A\setminus \widetilde A)}{\pi^{0}_{\Lambda}(A)}
\;\le\;
\sum_{\gamma:\,|\Lambda_\gamma|>\frac{1}{5}\delta L^{2}}
\exp\bigl(-(\beta-\lambda)|\gamma|\bigr).
\]

Applying the isoperimetric inequality together with the Peierls-type
argument used in \eqref{isoandpeier}, we conclude that the right-hand
side is bounded above by
\[
\exp\bigl(-c\,L^{\,a^{d(p)}-o(1)}\bigr)
\]
for some constant $c>0$. This yields the desired upper bound on
$\nu(\partial A_{p,a})$ and completes the proof of the claim.

By the isoperimetric inequality $|\Lambda_\gamma|\le |\gamma|^2/16$ and a Peierls–type estimate, we have
\[
\pi\big(\mathcal{C}_{\gamma,h}\big)
\ \le\
\exp\!\Big\{-\beta|\gamma|
\;+\;\,|\Lambda_\gamma|\,\hat\pi^p(\eta_0\ge h)
\;+\;c_0e^{-4\beta h}\,|\gamma|\,\log|\gamma|\Big\}.
\]
Invoking Lemma~\ref{lem:boundary-term} together with \eqref{eq:onepoint}, and for contours satisfying $|\Lambda_\gamma|\le \delta L^{2a^{d(p)}}$, we obtain
\[
\,|\Lambda_\gamma|\,\hat\pi^p(\eta_0\ge h)
\;\le\;
\,\delta\,L^{2a^{d(p)}}\cdot L^{-a^{d(p)}+\varepsilon}
\;=\;
\,\delta\,L^{\,a^{d(p)}+\varepsilon}.
\]
Choosing $\delta>0$ and $\varepsilon>0$ sufficiently small (and then $L$ sufficiently large), the perimeter term $-\beta|\gamma|$ dominates both the area correction and the logarithmic factor $e^{-4\beta h}\,|\gamma|\log|\gamma|$. Summing over admissible contours and applying a union bound yields
\begin{equation}\label{eq:bdry}
\nu(\partial A_{p,a})
\;\le\;
\exp\!\big\{-\,c_1\,L^{\,a^{d(p)}-o(1)}\big\},
\end{equation}
for some $c_1=c_1(\beta,p,a)>0$.

We next prove that
\begin{equation*}\label{eq:target}
    \nu(\mathcal B_a)\;\le\;\exp\bigl\{-c_2\,L^{\,a^{d(p)}}\bigr\}.
\end{equation*}
As in the proof of Lemma~\ref{probB}, the argument proceeds by separating the
contribution of short and long contours.

More precisely, we say that an $h$-contour $\gamma$ is \emph{short} if
\[
|\gamma|<\log^2 L,
\]
and \emph{long} otherwise.
We first show that
\[
\bar\pi^{p,0}_{\Lambda}\!\left(
\sum\nolimits_{\gamma}^{\prime}
\#\{x\in\Lambda_\gamma:\eta_x\ge h+1\}
\;\ge\;\tfrac12 L^2
\right)
\;\le\;
\exp\bigl(-c_2L^{2-o(1)}\bigr),
\]
where the summation $\sum_{\gamma}^{\prime}$ ranges over all short $h$-contours
$\gamma$.
As in the proof of Theorem~1, this bound follows from
\eqref{eq-h+1-given-C(gamma,h)}, together with a concentration argument.
In particular, after conditioning on the short contours, we estimate the
resulting sum by applying Hoeffding’s inequality.

We next consider the contribution of long contours.
We prove that there exists a constant $c_3>0$ such that
\[
\bar\pi^{p,0}_{\Lambda}\!\left(
\sum\nolimits_{\gamma}^{\prime\prime}
\#\{x\in\Lambda_\gamma:\eta_x\ge h+1\}
\;\ge\;\tfrac12 L^2
\right)
\;\le\;
\exp(-c_3L),
\]
where the summation $\sum_{\gamma}^{\prime\prime}$ is over all long $h$-contours
$\gamma$.
This estimate is obtained exactly as in the proof of Theorem~1, by combining a
Peierls-type argument with an isoperimetric bound on the contour length.

Combining \eqref{eq:bdry} and \eqref{eq:target} into \eqref{eq:std-decomp} yields
\[
\mathbb{P}_\nu\!\big(\tau_B<\exp\left\{\,c\,L^{\,a^{d(p)}}\right\}\big)
\;\le\;\exp\{\,c\,L^{\,a^{d(p)}-o(1)}\}\,\Big(\exp\left\{-c_1L^{\,a^{d(p)}+\epsilon}\right\}+\exp\left\{-c_2L^{\,a^{d(p)}}\right\}\Big)
\;=\;o(1).
\]
By monotonicity, the same lower bound holds for the chain started from the all-zero configuration.
Finally, since $\pi(\Omega_{p,a})\rightarrow 1$, we conclude that
\[
\mathbb{P}\!\big(\tau_{p,a}\ge \exp\{\,c\,L^{\,a^{d(p)}-o(1)}\}\big)\ \rightarrow\ 1\qquad(L\rightarrow\infty).
\]

\end{proof}

\section{Appendix}\label{appendix}
In this appendix we collect technical estimates and detailed computations that are used throughout the paper. In particular, several probabilistic bounds depend sensitively on the value of $p$, and verifying these requires a separate case-by-case analysis. To keep the main arguments streamlined, we defer these calculations here. We also provide complete proofs of auxiliary lemmas that were only stated in the main text.

\subsection{Proof of inequality~\eqref{eq:onepoint-tail}}\label{Appendix1}
We claim that for any $1<p<\infty$ and whenever $\log L \le n_+$,
\[
\bar{\pi}^p(\eta_0 \ge n_+/2) \;\le\; L^{-3}.
\]

This follows from known one-point large-deviation bounds available in the literature; for the reader’s convenience,
the relevant asymptotics are collected in
Table~\ref{tab:HpL}. Throughout, $o(1)$-terms are taken with respect to $\beta,n_+\rightarrow\infty$ and can be absorbed for large $L$ (with $\beta$ fixed and sufficiently large).

\smallskip
\noindent\textbf{Case $1<p<2$.}
From \cite[Theorem 5.1]{lubetzky2016harmonic}
\[
-\log \bar{\pi}^p(\eta_0 \ge h)=(c_p\beta+o(1))\,h^p,
\]
so
\[
\bar{\pi}^p(\eta_0 \ge n_+/2)
\;\le\; \exp\!\left(-\frac{c_p\beta+o(1)}{2^p}\,n_+^{\,p}\right)
\;\le\;\exp\!\left(-\frac{c_p\beta}{2^p}\,(\log L)^{p}+o\big((\log L)^{p}\big)\right).
\]
Since $p>1$ and $\beta$ is fixed and large, the exponent dominates $3\log L$ for large $L$, yielding the claim.

\smallskip
\noindent\textbf{Case $p=2$.}
From \cite[Theorem 3.1]{lubetzky2016harmonic}, we have
\[
-\log \bar{\pi}^2(\eta_0 \ge h)=(2\pi\beta+o(1))\,\frac{h^2}{\log h}.
\]
Hence
\[
\bar{\pi}^2(\eta_0 \ge n_+/2)
\;\le\;\exp\!\left(-\frac{2\pi\beta+o(1)}{4}\,\frac{n_+^{\,2}}{\log n_+}\right)
\;\le\;\exp\!\left(-\frac{\pi\beta+o(1)}{2}\,\frac{(\log L)^2}{\log\log L}\right)
\;\le\;L^{-3}
\]
for all large $L$, since $(\log L)^2/\log\log L \gg \log L$.

\smallskip
\noindent\textbf{Case $2<p<\infty$.}
From \cite[Theorem 5.5]{lubetzky2016harmonic}, we have
$$-\log \bar{\pi}^p(\eta_0 \ge h)\asymp \beta h^2.$$\\
So for some $c'_p>0$,
\[
\bar{\pi}^p(\eta_0 \ge n_+/2)\;\le\;\exp\!\left(-\frac{c'_p\beta}{4}\,n_+^{\,2}\right)
\;\le\;\exp\!\left(-\frac{c'_p\beta}{4}\,(\log L)^2\right)\;\le\;L^{-3}
\]
for all large $L$.

\smallskip
In all cases, the bound $\bar{\pi}^p(\eta_0 \ge n_+/2) \le L^{-3}$ holds for sufficiently large $L$ under the standing assumption that $\beta$ is sufficiently large (depending only on $p$).
\subsection{Proof of lemma~\ref{one-point-lemma}}\label{Appendix2}
\begin{proof}
We derive an upper bound on \(\hat{\pi}^{p}(\eta_0\ge a\,H(p,L)-1)\)
by combining known one-point large-deviation results from the literature, collected for convenience in Table~\ref{tab:HpL}. Throughout, \(o(1)\) denotes a quantity that vanishes as \(L\rightarrow\infty\), and \(a\in(0,1)\) is fixed.

\smallskip\noindent\textbf{Case 1: $1<p<2$.}
By the one–point LDP(\cite[Theorem 5.1]{lubetzky2016harmonic}) and the definition of $H=H(p,L)$ (\cite[Corollary 5.2]{lubetzky2016harmonic}),
\[
\hat\pi^p(\eta_0 \ge h)
=\exp\!\big(-(c_p\beta+o(1))\,h^p\big),
\qquad
\exp\!\big(-(c_p\beta+o(1))\,H^p\big)\ \ge\ \frac{5\beta}{L}\, .
\]
Setting $h=aH-1$ and noting that $(aH-1)^p\leq a^p H^p$, we obtain
\[
\hat\pi^p(\eta_0 \ge aH-1)
=\exp\!\big(-(c_p\beta+o(1))\,(aH-1)^p\big)
\ \le\ \exp\!\big(-(c_p\beta a^p+o(1))\,H^p\big)
\ \le\ \left(\frac{5\beta}{L}\right)^{a^p}.
\]
Hence, for any $\varepsilon>0$ and all sufficiently large $L$,
\[
\hat\pi^p\!\big(\eta_0 \ge aH(p,L)-1\big)\ \le\ L^{-a^p+\varepsilon}.
\]

\smallskip\noindent\textbf{Case 2: $p=2$.}
By the one–point LDP(\cite[Theorem 3.1]{lubetzky2016harmonic}) for the discrete Gaussian case and the definition of $H=H(2,L)$ (\cite[Theorem 2]{lubetzky2016harmonic}),
\[
\hat\pi^2(\eta_0 \ge h)
=\exp\!\left(-\,(2\pi\beta+o(1))\,\frac{h^2}{\log h}\right),
\qquad
(2\pi\beta)\,\frac{H^2}{\log H}\ \ge\ \log\!\frac{5\beta}{L}.
\]
For $h=aH-1$ we have $\log(aH-1)=\log H+\log\!\big(a-\tfrac{1}{H}\big)=\log H\,(1+o(1))$
and $(aH-1)^2=a^2H^2\,(1+o(1))$, hence
\[
\frac{(aH-1)^2}{\log(aH-1)}
=\frac{a^2H^2}{\log H}\,(1+o(1)).
\]
Therefore,
\[
\hat\pi^2(\eta_0 \ge aH-1)
=\exp\!\left(-\,(2\pi\beta+o(1))\,\frac{(aH-1)^2}{\log(aH-1)}\right)
=\exp\!\Big(-a^2(1+o(1))\,\log L\Big)
= L^{-a^2+o(1)}.
\]
Consequently, for any $\varepsilon>0$ and all sufficiently large $L$,
\[
\hat\pi^2\!\big(\eta_0 \ge aH(2,L)-1\big)\ \le\ L^{-a^2+\varepsilon}.
\]

\smallskip\noindent\textbf{Case 3: $2<p<\infty$.}
By the one-point LDP(\cite[Theorem 5.5]{lubetzky2016harmonic}) when $2<p<\infty$
Assume the two–sided bounds: for all sufficiently large $h$,
\begin{equation}\label{eq:two-sided}
C_1\,\beta h^2\ \le\ I(h):=-\log\hat\pi^p(\eta_0\ge h)\ \le\ C_2\,\beta h^2,
\end{equation}
for some constants $0<C_1\le C_2<\infty$. Choose the typical height $H=H(p,L)$ (\cite[Corollary 5.6]{lubetzky2016harmonic}) so that
\begin{equation}\label{eq:quantile-2p}
I\big(H(p,L)\big)\ =\ -\log\hat\pi^p(\eta_0\ge H)\ \le\ \log\!\frac{L}{5\beta}.
\end{equation}
From \eqref{eq:two-sided}–\eqref{eq:quantile-2p} it follows that
\[
\frac{1}{C_2\beta}\,\log L\,(1+o(1))\ \le\ H^2\ \le\ \frac{1}{C_1\beta}\,\log L\,(1+o(1)).
\]
Fix $a\in(0,1)$ and set $h=aH-1$. Then
\[
I(aH-1)\ \ge\ C_1\,\beta\,(aH-1)^2\ =\ C_1\,\beta\,a^2H^2\,\big(1+o(1)\big),
\]
and therefore
\begin{align*}
\hat\pi^p(\eta_0\ge aH-1)
&\le \exp\!\big(-C_1\beta a^2 H^2\,(1+o(1))\big)\\
&\le \exp\!\Big(-a^2\,\frac{C_1}{C_2}\,\log L\,\big(1+o(1)\big)\Big)
\;=\; L^{-a^2(C_1/C_2)+o(1)}.
\end{align*}
Here $o(1)\rightarrow0$ as $L\rightarrow\infty$.

Since the assumption \eqref{eq:two-sided} expresses $I(h)\asymp \beta h^2$ as $h\rightarrow\infty$, for any $\delta>0$ and all sufficiently large $h$ we may write
\[
(1-\delta)\,c_p\,\beta h^2\ \le\ I(h)\ \le\ (1+\delta)\,c_p\,\beta h^2.
\]
Taking $C_1=(1-\delta)c_p$ and $C_2=(1+\delta)c_p$ yields
\[
\frac{C_1}{C_2}\ \ge\ \frac{1-\delta}{1+\delta}\ \ge\ 1-2\delta.
\]
Letting $L$ be large (equivalently $H\rightarrow\infty$), choose $\delta=\delta(L)\rightarrow0$ slowly to obtain
\[
\hat\pi^p(\eta_0\ge aH-1)\ \le\ L^{-a^2(1-2\delta)+o(1)}.
\]
Given any $\varepsilon>0$, by taking $\delta$ and $o(1)$ smaller than $\varepsilon/4$ and $\varepsilon/2$, respectively, we conclude that
\[
\hat\pi^p(\eta_0\ge aH(p,L)-1)\ \le\ L^{-a^2+\varepsilon}
\]
for all sufficiently large $L$.
\end{proof}

\section*{Acknowledgements and Declarations}

\textbf{Acknowledgements}\\
This work was supported by the National Research Foundation of Korea (NRF) grant RS-2019-NR040050, and by Samsung Science and Technology Foundation under Project Number SSTF-BA2202-02.

\medskip

\textbf{Data availability}\\
Data sharing is not applicable to this article as no datasets were generated or analysed during the current study.

\medskip

\textbf{Declarations}\\
Conflicts of interest The author declares that he has no conflict of interest.

\bibliographystyle{plain}
\bibliography{iso}

\end{document}